\documentclass[12pt,a4paper]{article}

\usepackage{latexsym,amsmath,amscd,amssymb,amsthm,graphics}
\usepackage{enumerate,epsfig}
\usepackage[margin=2cm]{geometry}
 \usepackage{float}
\usepackage{graphicx}
\usepackage{framed,url,multicol}
\usepackage[usenames]{color}
\usepackage{shadow,enumerate}
\usepackage{makeidx}

\usepackage[colorlinks]{hyperref}

\usepackage[all]{xy}

\newcommand{\bfi}{\bfseries\itshape}

\makeatletter

\@addtoreset{figure}{section}
\def\thefigure{\thesection.\@arabic\c@figure}
\def\fps@figure{h, t}
\@addtoreset{table}{bsection}
\def\thetable{\thesection.\@arabic\c@table}
\def\fps@table{h, t}
\@addtoreset{equation}{section}

\makeatother

\newcommand{\bPi}{\boldsymbol{\Pi}}

\newcommand{\bOm}{\boldsymbol{\Omega}}

\newcommand{\bGam}{\boldsymbol{\Gamma}}

\newcommand{\bM}{\boldsymbol{M}}
\newcommand{\bN}{\boldsymbol{N}}

\newcommand{\bW}{\boldsymbol{\mathcal W}}
\newcommand{\W}{\boldsymbol{\mathcal W}}
\newcommand{\M}{\boldsymbol{\mathcal M}}
\newcommand{\V}{\boldsymbol{\mathcal V}}
\newcommand{\EP}{Euler--Poincar\'e }
\newcommand{\dede}[2]{\frac{\delta #1}{\delta #2}}
\newcommand{\prt}{\partial}

\newcommand{\ad}{\mbox{ad}}

\begin{document}

\newtheorem{theorem}{Theorem}[section]
\newtheorem{definition}[theorem]{Definition}
\newtheorem{lemma}[theorem]{Lemma}
\newtheorem{remark}[theorem]{Remark}
\newtheorem{proposition}[theorem]{Proposition}
\newtheorem{corollary}[theorem]{Corollary}
\newtheorem{example}[theorem]{Example}

\setlength\parindent{0pt}



\title{Matrix G-Strands}

\author{
Darryl D. Holm$^{1}$ and Rossen I. Ivanov$^{2}$ }
{\addtocounter{footnote}{1} \footnotetext{Department of
Mathematics, Imperial College, London SW7 2AZ, UK. Email:
\texttt{d.holm@imperial.ac.uk} } {\addtocounter{footnote}{1}
\footnotetext{ Department of Mathematical Sciences, Dublin
Institute of Technology, Kevin Street, Dublin 8, Ireland. Email:
\texttt{rossen.ivanov@dit.ie}
}%
%
\date{9 Feb 2014, clean copy}

\maketitle

\makeatother
\maketitle


%

\begin{framed}
\begin{abstract}
We discuss three examples in which one may extend integrable Euler--Poincar\'e ODEs to integrable Euler--Poincar\'e PDEs in the matrix G-Strand context. After describing matrix G-Strand examples for $SO(3)$ and $SO(4)$ we turn our attention to $SE(3)$ where the matrix G-Strand equations recover the exact rod theory in the convective representation. We then find a zero curvature representation (ZCR) of these equations and establish the conditions under which they are completely integrable. Thus, the G-Strand equations turn out to be a rich source of integrable systems. The treatment is meant to be expository and most concepts are explained in examples in the language of vectors in $\mathbb{R}^3$.
\end{abstract}
\end{framed}

\vspace{10 mm}

\tableofcontents

\section{Introduction and plan of the paper}\label{intro-sec}

This paper is about extending integrable Euler--Poincar\'e ODEs to integrable Euler--Poincar\'e PDEs in the matrix G-Strand context. The \EP theorem is the result of Hamilton's principle when reduction by symmetry is applied to  Lagrangians that are invariant under a Lie group. Poincar\'e introduced it in the context the rotating rigid body and the heavy top rotating rigidly under the force of gravity \cite{Po1901}, and his ideas were extended to continua in \cite{Ar1966,HoMaRa1998}.

To state the Euler--Poincar\'e equations, let $\mathfrak{g}$ be
a given Lie algebra and let $ \ell : \mathfrak{g} \rightarrow
\mathbb{R} $ be a given function (a La\-gran\-gian), let $\xi$ be an
element of $\mathfrak{g}$.
Then the evolution of the variable $\xi$ is determined by the
simplest form of the Euler--Poincar\'e equations, namely,
\begin{equation}
\frac{d}{dt} \frac{\delta \ell}{\delta \xi} =
\operatorname{ad}_{\xi}^{*}\frac{\delta \ell}{\delta \xi} \,,
\label{BasicEP-eqns}
\end{equation}
where $\delta \ell / \delta \xi \in \mathfrak{g}^\ast$ (the dual
vector space) is the variational derivative of $\ell$ with respect
to $\xi\in \mathfrak{g}$. The map $\operatorname{ad}_{\xi} :
\mathfrak{g} \to \mathfrak{g}$ is the linear map $\eta \mapsto
[\xi, \eta]$, where $[\xi, \eta ]$ denotes the Lie bracket of
$\xi$ and $\eta$, and $\operatorname{ad}_{\xi}^{*}:
\mathfrak{g}^{*} \to \mathfrak{g}^{*}$ is its dual (transpose) as
a linear map. These are the {\bfi basic Euler--Poincar\'e
equations}, and they are valid for either finite or infinite
dimensional Lie algebras.

We are interested in \emph{integrable} Euler--Poincar\'e equations, i.e., those that admit enough constants of motion to describe the solution as motion along their intersections.

\paragraph{Zero curvature representation (ZCR)}
Soon after the the KdV equation was written in commutator form by Lax \cite{Lax68}, people realized that the Lax-pair representation of integrable systems is equivalent to a zero curvature representation (a zero commutator of two operators). For example, Zakharov and Manakov \cite{ZaMa1973} developed the inverse scattering solution for the $3$-wave equation which has a $3 \times 3$ matrix ZCR representation, where both operators are linear in $\lambda$. Manakov used a particular $x$-independent case of the same ZCR \cite{Ma1976} to find the integrability conditions for the ordinary differential equations (ODEs) describing the motion of a rigid body in $n$ dimensions. Lax operators that are polynomial (quadratic) in $\lambda$  were first used in studies of the Thirring model in \cite{KuMi77} and DNLS equation in \cite{KaNe78}. The Inverse Scattering Method for Lax operators that are quadratic in $\lambda$ were developed in \cite{GeIvKu80} and recently a systematic Riemann-Hilbert formulation of the inverse-scattering problem for ZCR operators that are polynomial in $\lambda$ was developed in \cite{Ge12}.
%
%
%

Our starting point is the zero curvature representation (ZCR) of the Lax pair,
\begin{equation}
\partial_t L - \partial_s M -  [L,M] =0
\label{ZCR-eqn}
\,,\end{equation}
where $L$ and $M$ are matrix functions of the independent space-time variables $(s,t)$ that are quadratic in a constant spectral parameter $\lambda$,
\begin{equation}
L:=\lambda^2 A + \lambda \omega + \Gamma
\quad  \hbox{and} \quad
M:=\lambda^2 B + \lambda \gamma + \Omega
\,,
\label{LM-ops}
\end{equation}
and $(A,B)$ are nondegenerate constant mutually commuting matrices. Finding the ZCR is a sufficient condition for a space-time system of partial differential equations (PDE) to be  a completely integrable Hamiltonian system. An extensive development of the ZCR theory of integrable PDEs in $1+1$ dimensional space-time now exists. For example, the ZCR theory of integrable Hamiltonian equations extends, for example, to systems of hydrodynamic type \cite{DuNo1989}. See \cite{DuKrNo1985,ZaMaNoPi1984} for additional historical references and discussions of how to analytically obtain the coherent \emph{soliton} solutions of such equations.

The introduction of $1+1$ dimensional space-time dependence into the \EP equation \eqref{BasicEP-eqns} leads to the {\bfi G-Strand equations} \cite{Ho2011GM2,HoIvPe2012}. We shall show here that the G-Strand systems include integrable ZCR systems. The aim of the present paper is to extend the ideas of Poincar\'e, Arnold, Marsden and others to derive sufficient conditions for a G-Strand system to admit a ZCR, and to discuss a few good examples that illustrate the ideas. These examples will first extend the $SO(3)$ and $SO(4)$ rigid body ODE systems treated by Manakov in \cite{Ma1976} to $1+1$ PDE systems. Then we shall follow Poincar\'e's treatment of the heavy top \cite{Po1901} in treating G-Strands on $SE(3)$. The PDEs for the $SE(3)$ G-Strands turn out to recover the Simo-Marsden-Krishnaprasad (SMK) equations of exact rod theory in the convective representation of the dynamics of a flexible filament \cite{SiMaKr1988,EGBHPR2010}. The integrability conditions for the SMK equations are then discovered by writing their $SE(3)$ G-Strand form as a  ZCR, and solving the constraint relations that result.

\paragraph{Plan of the paper.}$\,$\\
The remainder of Section \ref{intro-sec} defines G-Strand equations and determines sufficient conditions on the group $G$ for which the G-Strand equations admit a zero curvature representation (ZCR).

Section \ref{spin-chains} treats the example of $SO(3)$ G-Strands and derives the conditions for which the motion of spin chains and $SO(3)$ chiral models admits a ZCR.

Section \ref{SO4-sec} extends the treatment of $SO(3)$ G-Strands in Section \ref{spin-chains} to the corresponding results for $SO(4)$ G-Strands, which contain the equations of classical $SO(4)$ chiral models.

Section \ref{SE3-sec} extends the treatment further by treating $SE(3)$ G-Strands, which contain the SMK equations of exact rod theory.

Section \ref{Conclus-sec} summarises and suggests future directions based on the present work.

In summary, the theory is outlined in the rest of Section \ref{intro-sec} and the remainder of the paper addresses examples.
The treatment of the examples is meant to be expository and most concepts are explained in the language of vectors in $\mathbb{R}^3$.

\subsection{Definitions}

\begin{enumerate}[(a)]
\item
Defining G-Strands

 \begin{definition}\rm
 A {\bf G-Strand} is a map $(s,t)\in\mathbb{R}\times\mathbb{R}$ into a Lie group $G$, $g(t,{s}):\,\mathbb{R}\times\mathbb{R}\to G$, whose dynamics in
$(s,t)$ may be obtained from
Hamilton's principle for a $G$-invariant reduced Lagrangian $\ell:
\mathfrak{g}\times\mathfrak{g}\to\mathbb{R}$, where $\mathfrak{g}$
is the Lie algebra of the group $G$.
The G-Strand system of hyperbolic partial differential equations for a $G$-invariant reduced Lagrangian consists of the Euler--Poincar\'e (EP) variational equations and an auxiliary compatibility equation.

\end{definition}

Subclasses of the G-Strand maps contain the principal chiral models of field theory in theoretical physics, reviewed, e.g., in  \cite{Wi1984,ZaMi1980}. An interpretation of the G-Strand equations as the dynamics of a continuous strand of oriented frames (or spins) is given in \cite{Ho2011GM2}.
This is the origin of the term, `strand'. The corresponding theory of molecular strands or filaments in three dimensions is discussed in  \cite{EGBHPR2010}. Recently, a covariant field theory of G-Strands in higher spatial dimensions (G-Branes) has also been developed \cite{FGB2012}.

\item
The G-Strand PDE system for variables $(\Omega,\Gamma)\in\mathfrak{g}\times\mathfrak{g}$ and Lagrangian $\ell(\Omega,\Gamma)$ is given by \cite{Ho2011GM2,HoIvPe2012}
\begin{align}
\begin{split}
{\partial_t} \frac{\delta \ell}{\delta \Omega}
&= {\rm ad}^*_\Omega\, \frac{\delta \ell}{\delta \Omega}
- \partial_s \frac{\delta \ell}{\delta \Gamma}
+ {\rm ad}^*_\Gamma\,\frac{\delta \ell}{\delta \Gamma}
\,,\\
\partial_t \Gamma
&= \partial_s\Omega
-  {\rm ad}_\Omega\,\Gamma
\,.
\end{split}
\label{GSeqns-EP+Compat}
\end{align}
Here ${\rm ad}_\Omega=[\Omega,\,\cdot\,]$ is the adjoint (ad) operation, which for matrix Lie algebras is given by the matrix commutator, and ${\rm ad}^*$ is its dual under a given pairing, such as the trace pairing.

The G-Strand equations on $\mathfrak{g}^*\times \mathfrak{g}$ shown in \eqref{GSeqns-EP+Compat} consist of a set of $n$ \EP equations in variables $(\Pi_1, \Pi_2,\dots,\Pi_n)\in \mathfrak{g}^*$, with $\Pi_k=\delta \ell/\delta \Omega_k$, and $n$ corresponding compatibility equations in variables $(\Omega_1, \Omega_2,\dots,\Omega_n)\in \mathfrak{g}$ that both involve only linear wave operators and quadratic nonlinear terms.

\begin{remark}\rm
The symbols $\Omega$ and $\Gamma$ appear in various guises throughout this paper. This is unavoidable because these symbols are standard for rotating bodies. However, this should cause no confusion, because the meaning will always be clear from the context, and the meaning of the symbols will be self-consistent within any particular section of the paper.
\end{remark}

\begin{remark}\rm
The G-Strand equations \eqref{GSeqns-EP+Compat} form a $2n\times2n$ symmetric hyperbolic system with constant characteristic speeds $c$ given by $c^2=1$. This property is a necessary condition for the matrix G-Strand systems to admit a ZCR.
\end{remark}

This paper aims to derive sufficient conditions for a G-Strand system to admit a ZCR, and to discuss a few classical examples that illustrate the ideas.

\item
Our considerations below will employ bi-invariant pairings on Lie algebras, so we will begin by reviewing a few of the relevant definitions and properties of such pairings. For more information about this topic, see \cite{Va1984}.
\begin{definition}\rm
A nondegenerate symmetric pairing $\langle\,\cdot\,,\,\cdot\,\rangle: \mathfrak{g}\times \mathfrak{g}\to\mathbb{R}$ is said to be \emph{bi-invariant} if it satisfies the ``associativity'' relation, for all $\eta,\,\xi,\,\zeta\in \mathfrak{g}$,
\begin{align}
\langle\,\eta\,,\,[\xi,\, \zeta]\,\rangle
=
\langle\,[\eta,\,\xi]\,,\,\zeta\,\rangle
\,,
\quad\hbox{or}\quad
\langle\,\eta\,,\,{\rm ad}_\xi \zeta\,\rangle
= \langle\,{\rm ad}_\eta\xi\,,\,\zeta\,\rangle
\,.
\label{bi-invar-def}
\end{align}
\end{definition}

In the following, we denote with musical symbols the maps $\flat:\mathfrak{g}\to\mathfrak{g}^*$ and $\sharp:\mathfrak{g}^*\to\mathfrak{g}$.

\begin{proposition}\label{bi-invar-dagger-ad}\rm
The condition
\begin{align}
{\rm ad}^\dagger_\xi\eta:=({\rm ad}^*_\xi\eta^\flat)^\sharp =-\,{\rm ad}_\xi\eta
\label{bi-invar-prop}
\end{align}
holds for Lie algebras with bi-invariant pairings.
\end{proposition}

\begin{proof}
Bi-invariance of the pairing allows us to set $\mathfrak{g}^*\simeq\mathfrak{g}$ and thereby suppress the $\flat$ and $\sharp$ notation in computing
\begin{align}
\begin{split}
\langle\,\eta\,,\,{\rm ad}_\xi \zeta\,\rangle
&=\langle\,{\rm ad}^*_\xi\eta\,,\,\zeta\,\rangle
\\&=\langle\,{\rm ad}^\dagger_\xi\eta\,,\,\zeta\,\rangle
\\
\langle\,\eta\,,\,{\rm ad}_\xi \zeta\,\rangle
&= \langle\,{\rm ad}_\eta\xi\,,\,\zeta\,\rangle
\\&= \langle\,-\,{\rm ad}_\xi\eta\,,\,\zeta\,\rangle
\quad\hbox{for all}\quad
\eta,\,\xi,\,\zeta\in \mathfrak{g}\,.
\end{split}
\label{proof-adstar}
\end{align}
The first line is the definition of the ${\rm ad}^*$ operation. The second line uses  the definition ${\rm ad}^\dagger_\xi\eta:=({\rm ad}^*_\xi\eta^\flat)^\sharp$, which becomes simply ${\rm ad}^\dagger_\xi\eta={\rm ad}^*_\xi\eta$ for a bi-invariant pairing. The third line repeats the definition of bi-invariance in  \eqref{bi-invar-def}, in preparation for the conclusion, which follows from antisymmetry of the ${\rm ad}$ operation. Hence,
\begin{align}
{\rm ad}^\dagger_\xi\eta={\rm ad}^*_\xi\eta=-\,{\rm ad}_\xi\eta
\label{addagger-adstar}
\end{align}
for Lie algebras with bi-invariant pairings.
\end{proof}

\begin{proposition}\rm
The condition ${\rm ad}^\dagger_\xi\eta=-\,{\rm ad}_\xi\eta$ holds on semi-simple matrix Lie groups.
\end{proposition}

\begin{proof}
Semi-simple matrix Lie groups have a bi-invariant pairing given by the negative of their corresponding Killing forms. Consequently, Proposition \ref{bi-invar-dagger-ad} implies the result.
\end{proof}

The following Theorem provides a sufficient condition for the \EP equation in the G-Strand system to be written solely in terms of adjoint (ad) operations, which in this case are matrix commutators.
\begin{theorem}
A sufficient condition for the \EP equation in the G-Strand system \eqref{GSeqns-EP+Compat} to be written solely in terms of ad operations is that
\begin{equation}
{\rm ad}_\xi^\dagger\eta =-\, {\rm ad}_\xi\eta
\label{ad-adstar}
\,,\end{equation}
which holds for any Lie algebra with a bi-invariant pairing.
\end{theorem}

\begin{proof}
Denote $\frac{\delta \ell}{\delta \Omega}=:\omega^\flat$ and $\frac{\delta \ell}{\delta \Gamma}=:-\,\gamma^\flat$, so the \EP equation in \eqref{GSeqns-EP+Compat} becomes
\begin{align}
\begin{split}
{\partial_t} \omega^\flat
&= {\rm ad}^*_\Omega\, \omega^\flat
+ \partial_s \gamma^\flat
- {\rm ad}^*_\Gamma\,\gamma^\flat
\\
{\partial_t} \omega
&= ({\rm ad}^*_\Omega\, \omega^\flat)^\sharp
+ \partial_s \gamma
- ({\rm ad}^*_\Gamma\,\gamma^\flat)^\sharp
\\
&= {\rm ad}^\dagger_\Omega\, \omega
+ \partial_s \gamma
- {\rm ad}^\dagger_\Gamma\,\gamma
\\
&= -\,{\rm ad}_\Omega\, \omega
+ \partial_s \gamma
+ {\rm ad}_\Gamma\,\gamma
\,.
\end{split}
\label{GSeqns-EP}
\end{align}
which is written solely in terms of ad operations.
\end{proof}

\begin{corollary}
G-Strand systems \eqref{GSeqns-EP+Compat} on semi-simple matrix Lie groups may be written entirely in terms of matrix derivatives and commutators.
\end{corollary}

\begin{proof}
As a consequence of \eqref{GSeqns-EP}, the G-Strand system \eqref{GSeqns-EP+Compat} may be written for semi-simple matrix Lie groups equivalently as
\begin{align}
\begin{split}
{\partial_t} \omega
&= -\,{\rm ad}_\Omega\,\omega
+ \partial_s \gamma
+ {\rm ad}_\Gamma\,\gamma
\\
&= -\,[\Omega,\,\omega]
+ \partial_s \gamma
+ [\Gamma,\,\gamma]
\,,\\
\partial_t \Gamma
&= \partial_s\Omega
-  {\rm ad}_\Omega\,\Gamma
\\
&= \partial_s\Omega
-  [\Omega,\,\Gamma]
\,,
\end{split}
\label{GSeqns-com}
\end{align}
which involves only matrix derivatives and commutators.
\end{proof}

\end{enumerate}

\subsection{Lax pairs for matrix G-Strand systems}

\begin{theorem} \rm [ZCR formulation for matrix G-Strands] \label{ZCR-thm}$\,$\\
The matrix G-Strand system of equations in commutator form \eqref{GSeqns-com}, may be expressed as a zero curvature representation (ZCR, or Lax pair) on the Lie algebra $\mathfrak{g}$ of the Lie group $G$,
\begin{equation}
\partial_t L - \partial_s M -  [L,M] =0
\label{ZCR-eqn}
\,,\end{equation}
with
\begin{equation}
L:=\lambda^2 A + \lambda \omega + \Gamma
\quad  \hbox{and} \quad
M:=\lambda^2 B + \lambda \gamma + \Omega
\,,
\label{LM-ops}
\end{equation}
where the independent constant matrices $A$ and
$B$ commute, ${\rm ad}_AB=[A,B]=0$, and
\begin{equation}
{\rm ad}_A \gamma =  {\rm ad}_B\, \omega \quad  \hbox{and} \quad
{\rm ad}_A \Omega =  {\rm ad}_B \Gamma - {\rm ad}_{\omega}\gamma
\,. \label{LM-ops}
\end{equation}

\end{theorem}
\begin{proof}
Inserting the definitions for L and M, then equating coefficients to zero at each power of $\lambda$ yields:
\begin{align}
\begin{split}
\lambda^4 &: [A,B]=0 \\
\lambda^3 &: [A,\gamma]-[B,\omega]=0 \\
\lambda^2 &: [A,\Omega]-[B,\Gamma]+[\omega,\gamma]=0 \\
\lambda^1 &: -\, \left[\Omega,\omega\right] + \left[\Gamma,\gamma\right] = \partial_t \omega -\partial_s \gamma \\
\lambda^0 &:  \left[\Omega,\Gamma\right] = \partial_t \Gamma - \partial_s \Omega
\end{split}
\label{constraint-rels}
\end{align}
We may now solve the relations in \eqref{constraint-rels} and extract the conditions under which the matrix G-Strand equations will possess a ZCR (Lax pair) and thus will be an integrable system.

\begin{itemize}
\item The equations at order $\lambda^0$ and $\lambda^1$ in the
system \eqref{constraint-rels} recover the compatibility condition
and Euler--Poincar\'e equations, respectively. These equations
determine the time evolution of $\omega$ and $\Gamma$, which are
the \emph{prognostic} variables appearing in $L$. In contrast, the
variables $\gamma$ and $\Omega$ appearing in $M$ are
\emph{diagnostic} variables determined
algebraically from the prognostic variables after each time step.
\item The order $\lambda^3$ equation gives a linear relation that
determines the  diagnostic variable $\gamma$ in terms of the
prognostic variable $\omega$, namely,
\begin{align}
{\rm ad}_A\gamma = {\rm ad}_B\, \omega \label{lambda3-ad}
\end{align}
A similar condition $\gamma = {\rm ad}_A^{-1}{\rm ad}_B\, \omega$
occurs in Manakov's formulation of the Lax pair for the rigid
body on $SO(n)$ \cite{Ma1976}, interpreted as a linear relation
between the angular momentum and angular velocity that restricts
the moments of inertia for which the $SO(n)$ rigid body is
integrable. However, in our considerations below
${\rm ad}_A$ is not an invertible operator. This means that elements of the 
kernel of ${\rm ad}_A$ will appear, multiplied by scalars that must be
determined in finding the Hamiltonians for integrable cases of the G-strands.  

\item The order $\lambda^2$ equation gives \emph{another} formula
that determines the \emph{other} diagnostic variable $\Omega$ in
terms of the prognostic variables $\omega$ and $\Gamma$, namely,
\begin{align}
{\rm ad}_A \Omega =  {\rm ad}_B \Gamma - {\rm ad}_{\omega}\gamma.
\label{lambda2-ad}
\end{align}
\end{itemize}
Thus,  the constraint equations \eqref{constraint-rels} may be
solved to express the diagnostic variables ($\gamma,\Omega$)
contained in $M$,  in terms of the prognostic ones
($\omega,\Gamma$) contained in $L$. Consequently, the $G$-Strand
system of equations \eqref{GSeqns-com} has a ZCR (or Lax pair) as
in \eqref{ZCR-eqn}.

%
\end{proof}
Theorem \ref{ZCR-thm} states that under certain conditions the G-Strand equations on semi-simple Lie algebras possess a ZCR, given by
\begin{equation}
\partial_t L - \partial_s M -  [L,M] =0
\label{ZCR-eqn-intro}
\,,\end{equation}
with
\begin{equation*}
L:=\lambda^2 A + \lambda \omega + \Gamma
\quad  \hbox{and} \quad
M:=\lambda^2 B + \lambda \gamma + \Omega
\,,
\end{equation*}
where $A$ and $B$ satisfy ${\rm ad}_AB=0$. The
ZCR, in turn, puts the matrix G-Strands into the realm of
integrable Hamiltonian systems. One can pursue soliton solutions
of the matrix G-Strand equations derived here, for example, by
following the dressing method for ZCRs \cite{ZaMaNoPi1984}.

The remainder of the paper consists of explicit
examples of this procedure that directly extend integrable
Euler--Poincar\'e ODEs to integrable Euler--Poincar\'e PDEs in the
matrix G-Strand context by computing the ZCR's for the three Lie
groups, $SO(3)$, $SO(4)$ and $SE(3)$.

The theory in this section and the examples in the remainder of the paper make it clear that the G-Strand equations are a rich source of integrable systems.

\paragraph{Relation to complex fluids (CF)}
After performing the Legendre transformation of the Lagrangian in the G-Strand equations \eqref{GSeqns-EP+Compat},
\[
h(\Pi, \Gamma) = \langle \,\Pi,\, \Omega\,\rangle - \ell(\Omega,\,\Gamma)
\,,\]
one computes the derivatives of the Hamiltonian $h$ as
\[
\frac{\delta h}{\delta \Pi} = \Omega
\,,\quad
\frac{\delta h}{\delta \Omega}
= 0 =
\Pi - \frac{\delta \ell}{\delta\Omega}
\,,\quad
\frac{\delta h}{\delta \Gamma} = -\,\frac{\delta \ell}{\delta \Gamma}
\,.\]
Inserting these relations into the G-Strand equations \eqref{GSeqns-EP+Compat} yields the following Hamiltonian equations,
\begin{align}
\begin{split}
{\partial_t} \Pi
&= {\rm ad}^*_{\delta h/\delta \Pi}\, \Pi
+\left( \partial_s
- {\rm ad}^*_\Gamma\right)\frac{\delta h}{\delta \Gamma}
\,,\\
\partial_t \Gamma
&= \left( \partial_s + {\rm ad}_\Gamma\right)
\frac{\delta h}{\delta \Pi}
\,.
\end{split}
\label{GSeqns-Ham+Compat}
\end{align}
or, in matrix form,
\begin{equation}
\frac{\partial}{\partial t}
    \begin{bmatrix}
    \Pi
    \\
    \Gamma
    \end{bmatrix}
=
\begin{bmatrix}
  -{\rm ad}^\ast_{\Pi}
   &
   (\partial_s - {\rm ad}^*_\Gamma)
   \\
   (\partial_s + {\rm ad}_\Gamma)
   & 0
    \end{bmatrix}
    \begin{bmatrix}
   \delta h/\delta\Pi \\
   \delta h/\delta\Gamma
    \end{bmatrix}.
 \label{CF-HamMatrix}
\end{equation}


Remarkably, equations \eqref{GSeqns-Ham+Compat} are exactly the
Hamiltonian equations for a static CF whose {\bfi broken symmetry}
is $G$ \cite{Ho2002,GBRa2009}.

Being dual to a Lie algebra, the matrix in equation
\eqref{CF-HamMatrix} defines a {\bfi Lie--Poisson Hamiltonian
matrix}. See, e.g., \cite{MaRa1999, Ho2011GM2} and references
therein for more discussions of such Hamiltonian matrices. For our
present purposes, its emergence in the {CF} context links the
physical and mathematical interpretations of the variables in the
theory of {CF}s with earlier work in the gauge theory approach to
condensed matter, see, e.g., \cite{Kleinert1989}.
These gauge theory aspects may be recognised by noticing that the $\Pi$-$\Gamma$ cross terms are the {\bfi covariant derivatives} with respect to the space-time connection one-form given by $\Omega \,dt + \Gamma ds$, as explained in \cite{HoKu1988}.  The second G-Strand equation in \eqref{GSeqns-Ham+Compat} is then recognized as the zero curvature relation for this connection one-form.

The gauge theory approach to liquid crystal
physics is reviewed in, e.g.,  \cite{Tr1982,Kleinert1989,Kl1983,Kl1989}.
In the present situation, we may identify $\Pi$ as the gauge charge and $\Gamma$ as the vector potential for the gauge theory of condensed matter. The same identifications apply to \emph{chromohydrodynamics}, the dynamics of a Yang-Mills fluid plasma \cite{GiHoKu1982}. These identifications also apply to spin glasses and superfluid $^4He$ and $^3He$-$A$ with rotation and spin \cite{HoKu1982,HoKu1988}. In addition, they apply to the theory of complex fluids, such as liquid crystals. In fact the theories of superfluids and complex fluids are very similar, when examined from this viewpoint \cite{KlMi1978,Mi1980}.

The partial derivatives $\partial_s$ in equation \eqref{CF-HamMatrix} appearing in the $\Pi$-$\Gamma$ cross terms comprise a generalized two-cocycle. This is somewhat exotic for a classical fluid. Finding such a feature in the continuum theory of complex fluids may seem fitting, because in their shared Hamiltonian structures the complex fluids seem to form a bridge between the classical and quantum fluid theories, such as superfluids.
Essentially all of the complex fluid types can be formulated naturally as \EP systems using the methods of \cite{Ho2002,GBRa2009,EGBHPR2010}.
This will be the topic of future research.

\section{$SO(3)$ G-Strands}\label{spin-chains}

One possibility for the construction of an $SO(3)$ integrable
G-Strand system was analyzed in \cite{HoIvPe2012} by linking it to
the integrable $P$-chiral model of \cite{BoYa1995,Ya1988}. The present
section will derive a zero curvature representation leading
to a \emph{new} class of integrable $SO(3)$ G-Strand equations.

\subsection{The hat map $\,\mathbf{\widehat{\,}}\,:\,({so}(3),
[\cdot, \cdot]) \to (\mathbb{R}^3, \times)$}

The Lie algebra $(\mathfrak{so}(3), [\cdot, \cdot])$ with matrix commutator bracket $[\,\cdot\,,\,\cdot\,]$ maps to the Lie algebra $(\mathbb{R}^3, \times) $ with vector product $\times$, by the linear isomorphism
\begin{eqnarray*}
\label{so three isomorphism in coordinates}
\mathbf{u}: =(u^1, u^2, u ^3) \in \mathbb{R}^3 \mapsto
\widehat{u}: =
\left[
\begin{array}{ccc}
0&-\,u^3&u^2\\
u^3&0&-\,u^1\\
-\,u^2&u^1&0
\end{array}
\right] \in {so}(3)
\,.
\end{eqnarray*}
In matrix and vector components, the linear isomorphism is
$
\widehat{u}_{ij}: =-\,\epsilon_{ijk}u^k
\,.
$
Equivalently, this isomorphism is given by
$
\widehat{u} \mathbf{v} = \mathbf{u}\times \mathbf{v} \quad
\text{for all} \quad \mathbf{u}, \mathbf{v}\in \mathbb{R}^3.
$
This is the hat map $\,\mathbf{\widehat{\,}}\,:\,({so}(3),
[\cdot, \cdot]) \to (\mathbb{R}^3, \times)$, which holds for the skew-symmetric $3\times3$ matrices in the matrix Lie algebra $\mathfrak{so}(3)$.

One may verify the following useful formulas for $\mathbf{u}, \mathbf{v}, \mathbf{w} \in \mathbb{R}^3$:
\begin{eqnarray*}
\label{Lie-bracket-relation}
(\mathbf{u} \times \mathbf{v})\widehat{\phantom{u}}
&=&
{\widehat{u}}\,{\widehat{v}} - {\widehat{v}}\,{\widehat{u}}
=:[{\widehat{u}},{\widehat{v}}]
\,,
\nonumber\\
\label{triple-cross-product}
[{\widehat{u}},{\widehat{v}}]\,\mathbf{w}
&=&
(\mathbf{u} \times \mathbf{v}) \times \mathbf{w}
\,,\\
\label{double-cross-product}
\big((\mathbf{u} \times \mathbf{v}) \times \mathbf{w}\big)
\widehat{\phantom{u}}
&=&
\big[[{\widehat{u}},{\widehat{v}}]
\,,\, {\widehat{w}}
\big]
\,,
\\
\mathbf{u}\cdot \mathbf{v}
&=& - \tfrac{1}{2}
\operatorname{trace}({\widehat{u}}\, {\widehat{v}})
=:\big\langle \,{\widehat{u}}
\,,\, {\widehat{v}}\, \big\rangle
\,,
\label{dot-product}
\end{eqnarray*}
in which the dot product of vectors is also the natural pairing of $3\times3$
skew-symmetric matrices.

\subsection{The $SO(3)$ G-Strand system in $\mathbb{R}^3$ vector form}

By using the hat map, $\mathfrak{so}(3)\to\mathbb{R}^3$, the matrix G-Strand system for $SO(3)$ \cite{HoIvPe2012} may be written in $\mathbb{R}^3$ vector form by following the analogy with the Euler rigid body \cite{Ma1976} in standard notation, cf. equations \eqref{GSeqns-com}
\begin{align}
\begin{split}
\partial_t \Pi + \Omega\times\Pi  - \partial_s \Xi - \Gamma\times \Xi &= 0
\,,
\\
\partial_t \Gamma - \partial_s \Omega - \Gamma\times \Omega &= 0
\,,
\end{split}
\label{EPeqn+compat-SO3}
\end{align}
where $\Omega:=O^{-1}\partial_t O\in\mathfrak{so}(3)$ and
$\Pi:=\delta \ell/\delta \Omega\in\mathfrak{so}(3)^*$ are the body
angular velocity and momentum, while
$\Gamma:=O^{-1}\partial_sO\in\mathfrak{so}(3)$ and $\Xi=-\,\delta
\ell/\delta \Gamma\in\mathfrak{so}(3)^*$ are the body angular
strain and stress. These G-Strand equations for
$\mathfrak{g}=\mathfrak{so}(3)$ equations may be expressed in
Lie--Poisson Hamiltonian form in terms of vector operations in
$\mathfrak{se}(3)\simeq\mathfrak{so}(3)\times\mathbb{R}^3$ as,
\begin{equation}
\frac{\partial}{\partial t}
    \begin{bmatrix}
    \Pi
    \\
    \Gamma
    \end{bmatrix}
=
\begin{bmatrix}
  \Pi\times
   &
   \partial_s + \Gamma\times
   \\
    \partial_s + \Gamma\times
   & 0
    \end{bmatrix}
    \begin{bmatrix}
   \delta h/\delta\Pi = \Omega\\
   \delta h/\delta\Gamma =  \Xi
    \end{bmatrix}
    .
    \label{LP-Ham-struct-SO3}
\end{equation}
This Hamiltonian matrix yields a {\bfi Lie--Poisson bracket} defined on the dual of the semidirect-product Lie algebra $\mathfrak{se}(3)\simeq\mathfrak{so}(3)\circledS\,\mathbb{R}^3$ with a two-cocycle given by $\partial_s$. Namely,
\begin{align}
\begin{split}
\{ f,\,h\} &= \int
    \begin{bmatrix}
   \delta f/\delta\Pi\,, &
   \delta f/\delta\Gamma
    \end{bmatrix}
    \cdot
\begin{bmatrix}
  \Pi\times
   &
   \partial_s + \Gamma\times
   \\
    \partial_s + \Gamma\times
   & 0
    \end{bmatrix}
    \begin{bmatrix}
   \delta h/\delta\Pi\\
   \delta h/\delta\Gamma
    \end{bmatrix}
    ds
\\  &\hspace{-1cm}=  \int
-\,\Pi\cdot  \frac{\delta f}{\delta\Pi} \times  \frac{\delta h}{\delta\Pi}
-\Gamma \cdot \left(
\frac{\delta f}{\delta\Pi} \times\frac{\delta h}{\delta\Gamma}
-
\frac{\delta h}{\delta\Pi} \times\frac{\delta f}{\delta\Gamma}
\right)
+
\frac{\delta f}{\delta\Pi} \partial_s\frac{\delta h}{\delta\Gamma}
+
\frac{\delta f}{\delta\Gamma} \partial_s\frac{\delta h}{\delta\Pi}
\    ds
    .
\end{split}
\label{LP-brkt-SO3}
\end{align}

Dual variables are $\Pi$ dual to $\mathfrak{so}(3)$ and $\Gamma$ dual to $\mathbb{R}^3$. For more information about Lie--Poisson brackets, see \cite{MaRa1999}.

The $\mathbb{R}^3$ G-Strand equations \eqref{EPeqn+compat-SO3} combine two classic ODEs due separately to Euler and Kirchhoff into a single PDE system.
The $\mathbb{R}^3$ vector representation of $\mathfrak{so}(3)$ implies that ${\rm ad}^*_\Omega\Pi= -\,\Omega\times\Pi =-{\rm ad}_\Omega\Pi$, so the corresponding Euler--Poincar\'e equation has a ZCR.
To find its integrability conditions, we set
\begin{equation}
L:=\lambda^2 A + \lambda \Pi + \Gamma
\quad  \hbox{and} \quad
M:=\lambda^2 B + \lambda \Xi + \Omega
\,,
\label{LMdef-SO3}
\end{equation}
and compute the conditions in terms of $\Pi$ $\Omega$, $\Xi$, $\Gamma$ and the constant vectors $A$ and $B$ that are required to write the vector system \eqref{EPeqn+compat-SO3} in zero-curvature form,
\begin{equation}
\partial_t L - \partial_s M -  L\times M =0
\label{ZCRvec-eqn}
\,.\end{equation}
By direct substitution of \eqref{LMdef-SO3} into  \eqref{ZCRvec-eqn} and equating the coefficient of each power of $\lambda$ to zero, one finds
\begin{align}
\begin{split}
\lambda^4 &: A\times B=0   \\
\lambda^3 &: A\times \Xi - B\times \Pi = 0 \\
\lambda^2 &: A\times \Omega - B\times \Gamma + \Pi\times\Xi =0  \\
\lambda^1 &:  \Pi\times\Omega + \Gamma\times \Xi
= \partial_t \Pi  - \partial_s \Xi
\quad\hbox{(EP equation)} \\
\lambda^0 &:  \Gamma\times \Omega = \partial_t \Gamma - \partial_s \Omega
\quad\hbox{(compatibility)}
\end{split}
\label{constraint-rels-SO3}
\end{align}
where $A$ and $B$ are taken as constant nonzero vectors.
These imply the following relationships
\begin{align}
\begin{split}
\lambda^4 &: A=\alpha B  \\
\lambda^3 &: A\times(\Xi-\Pi/\alpha)  = 0
\Longrightarrow \Xi - \Pi/\alpha = \beta A \\
\lambda^2 &: A \times(\Omega - \Gamma/\alpha) = \Xi \times\Pi
= \beta A \times\Pi
\\
\end{split}
\label{constraint-rel-imp}
\end{align}

We solve equations \eqref{constraint-rel-imp} for the diagnostic variables $\Xi$ and $\Omega$, as
\begin{align}
\Xi - \frac{1}{\alpha}\Pi = \beta A
\quad\hbox{and}\quad
\Omega - \frac{\Gamma}{\alpha} - \beta\Pi = \gamma A\,,
\label{diagnos-relations}
\end{align}
where $\alpha,\beta,\gamma$ are real scalar functions.
Hence, we have proved,
\begin{theorem} \rm [ZCR formulation for $SO(3)$ G-Strand] \label{ZCR-thm-SO3}$\,$\\
The $SO(3)$ matrix G-Strand system of equations in \eqref{EPeqn+compat-SO3}, may be expressed as a zero curvature representation (ZCR, or Lax pair),
\begin{equation}
\partial_t L - \partial_s M -  L\times M =0
\label{ZCRvec-eqn1}
\,,\end{equation}
with
\begin{equation}
L:=\lambda^2 A + \lambda \Pi + \Gamma
\quad  \hbox{and} \quad
M:=\lambda^2 B + \lambda \Xi + \Omega
\,,
\label{LMdef-SO3a}
\end{equation}

where $A$ and $B$ are constant nonzero vectors, and
the diagnostic variables $\Xi$ and $\Omega$, are given as in \eqref{diagnos-relations}.
\end{theorem}

\begin{remark}\rm
The Lax pair in \eqref{LMdef-SO3a} is quadratic in $\lambda$. The Lax pair formulation of the heavy top in \cite{Ho2011GM2} is also quadratic in $\lambda$.
\end{remark}

\subsection{Conserved quantities for the SO(3) $G$-strand}

In general, the Lax pair representation  of an integrable system determines its conserved quantities (integrals of motion). This section
presents the computation of the conserved quantities for
the SO(3) $G$-strand. In the $\mathfrak{so}(3)$ basis of Pauli matrices, $\sigma^k$ with $k=1,2,3$, the $L$-operator in \eqref{LMdef-SO3a} takes the following form
\[L = \left(\begin{array}{cc}
L_3 & L_1-iL_2 \\ L_1+iL_2 &  -L_3
\end{array}\right)
=
\sum_{k=1}^3L_k\sigma_k,\]
in which $L_k=\lambda^2 A_k+\lambda \Pi_k +
\Gamma _k $, with $k=1,2,3$.  The $M$-operator takes a similar form in this basis. The notation may be shortened further by introducing $L_{12}:=
L_1-iL_2 = L_{21}^*$. The Lax representation is the compatibility relation for two
linear problems,

\begin{equation}
\psi_s+L\psi=0, \qquad  \psi_t+M\psi=0,\label{Lax pair}
\end{equation} where $\psi=(\psi_1,\psi_2)^T$ is a two-component
vector. The first equation in (\ref{Lax pair}) can be rewritten equivalently
in terms of the first component $\psi_1$ only, by eliminating $\psi_2$ to find
\begin{equation}
\psi_{1,ss}= \left(L_1^2+L_2^2+L_3^2-L_{3,s}+L_3
\frac{L_{12,s}}{L_{12}}\right)\psi_1+\frac{L_{12,s}}{L_{12}}\psi_{1,s}\label{Lax
pair2}
\end{equation}

Compatibility of the density $\tilde{\rho}:=\frac{\psi_{1,s}}{\psi_1}$ and  flux $\tilde{\mu}:=\frac{\psi_{1,t}}{\psi_1}$ yields the local conservation law
\[\tilde{\rho}_t=\tilde{\mu}_s.\]
Moreover, equation (\ref{Lax pair2}) implies that the quantity
\[
\rho=\ln(\psi_1\sqrt{L_{12}})
\]
is \emph{also} a conserved density which satisfies the Riccati equation
\begin{equation}
L_{12}^2(\rho_s+\rho^2)=L_{12}^2
\left(L_1^2+L_2^2+L_3^2\right)-L_3L_{12}L_{12,s}-\frac{1}{2}L_{12}L_{21,ss}+L_{12,s}^2
\,.\label{Ric}
\end{equation}

Such a Riccati equation admits an infinite power series solution in the form,
\[
\rho=|A|\lambda^2+\rho_{-1}\lambda+\rho_0+\frac{\rho_1}{\lambda}+\frac{\rho_2}{\lambda^2}+\ldots
\]
in which $\rho_n$, $n=-1,0,1,\ldots$ are the densities of
the infinitely many conserved quantities for the integrable system implied by compatibility of the linear equations in the Lax pair \eqref{Lax pair}. These can be determined
recursively by comparing the coefficients in front of the equal
powers of $\lambda$ in (\ref{Ric}), while keeping in mind that the $L_k$ are
quadratic in $\lambda$. For example, balancing the
coefficients of $\lambda^7$ gives the first conserved
density (in $\mathbb{R}^3$ vector notation),
\begin{equation}
\rho_{-1}= \frac{A\cdot \Pi}{|A|}\label{Int1}
.
\end{equation}

Likewise, balancing the coefficients of
$\lambda^6$ and $\lambda^5$  gives, respectively,
\begin{equation}
\begin{split}
\rho_0=&\frac{1}{|A|}\left(\frac{(A\times\Pi)^2}{2|A|^2}+A\cdot\Gamma\right)
\,,\\
\rho_1=&\frac{1}{|A|}\left(\Pi\cdot\Gamma-\rho_{-1}\rho_0-\frac{1}{2}\rho_{-1,s}-
\left(\frac{A_3(\Pi_1-i\Pi_2)}{A_1-iA_2}\right)_s\right).
\end{split}
\end{equation}

Recalling that $|A|$ is a constant and neglecting total derivatives
yields an equivalent $\rho_1$ density
\begin{equation}
\rho_{1}^*= \frac{1}{|A|}\left(\Pi\cdot\Gamma
- \frac{(A\cdot\Pi)}{|A|^2}
\left(
 \frac{(A\times\Pi)^2}{2|A|^2}+A\cdot\Gamma\right)\right)
 \label{Int3}
.\end{equation}

The conserved quantities arising from the Lax representation may now be evaluated as
\begin{equation}
\begin{split}
H_{-1}=& \int (A\cdot\Pi)d s
\,,\\
H_0=&\int\left(\frac{(A\times\Pi)^2}{2|A|^2}+A\cdot\Gamma\right)d
s
\,,\\
H_1=&\int\left(\Pi\cdot\Gamma - \frac{(A\cdot\Pi)}{|A|^2}
\left(
 \frac{(A\times\Pi)^2}{2|A|^2}+A\cdot\Gamma\right)\right)ds
 \,.
\end{split}
\end{equation}

We shall seek the Hamiltonian $h$ in the Lie-Poisson form of the G-strand equations in \eqref{LP-Ham-struct-SO3} as a linear combination
of $H_{-1}$, $H_{0}$ and $H_{1}$. That is, we take
$h=c_{-1}H_{-1}+c_0H_0+c_1H_1$ for numerical constants $c_k$ yet to be determined.
The Hamiltonian $h$ must satisfy the two relations

\begin{equation}\label{Hamilt rel}
\begin{split}
\frac{\delta h}{\delta \Pi}=&c_{-1}\frac{\delta H_{-1}}{\delta
\Pi}+ c_{0}\frac{\delta H_{0}}{\delta \Pi}+c_{1}\frac{\delta
H_{1}}{\delta \Pi}=
\frac{1}{\alpha}\Gamma+\beta\Pi+\gamma A=: \Omega,\\
\frac{\delta h}{\delta \Gamma}=&c_{-1}\frac{\delta H_{-1}}{\delta
\Gamma}+ c_{0}\frac{\delta H_{0}}{\delta \Gamma}+c_{1}\frac{\delta
H_{1}}{\delta \Gamma}= \frac{1}{\alpha}\Pi+\beta A =: \Xi
\end{split}
\end{equation}

Comparing the scalar coefficients of the vectors
$A$, $\Pi$ and $\Gamma$ in (\ref{Hamilt rel})
with the diagnostic relations for $\Omega$ and $\Xi$ in \eqref{diagnos-relations} yields
\begin{equation}\label{Hamilt rel 1}
\begin{split}
\alpha=&c_{1}=1,\\
\beta=&c_{0}-\frac{A\cdot\Pi}{|A|^2},\\
\gamma=&-\frac{1}{|A|^2}\left(\frac{\Pi^2}{2}+A\cdot\Gamma\right)
+c_{-1}-c_0\frac{A\cdot\Pi}{|A|^2}+\frac{3(A\cdot\Pi)^2}{2|A|^4} ,
\end{split}
\end{equation} $c_{-1}$ and $c_0$ are arbitrary real
constants. The most general Hamiltonian could in principle contain
linear combinations of all the conserved quantities, with
coefficients $c_k$ where possibly $k>1$. In such cases the
expressions for $\alpha$, $\beta$ and $\gamma$ will contain terms
related to the higher conserved quantities entering the
Hamiltonian. In fact, it is not difficult to spot that a possible
more general choice for a Hamiltonian is

\begin{equation}\label{Hamilt so3 general}
h=\int \left(\Pi\cdot \Gamma + (\mu A\cdot\Pi +
\nu)\left(\frac{1}{2} \Pi^2 + A\cdot \Gamma\right) + f( A\cdot\Pi)
\right) ds
\end{equation}

where $\mu, \nu$ are arbitrary constants and $f( A\cdot\Pi)$ is an
arbitrary function of argument $ A\cdot\Pi$.  In such case the
scalars from (\ref{diagnos-relations}) are

\begin{equation}\label{Hamilt rel 2}
\begin{split}
\alpha=&1,\\
\beta=&\mu A\cdot\Pi +\nu,\\
\gamma=&\mu\left(\frac{1}{2} \Pi^2 + A\cdot
\Gamma\right)+f'(A\cdot\Pi) ,
\end{split}
\end{equation}

where $f'$ is the first derivative of $f$.

\paragraph{Lagrangian.}
The relation between the Lagrangian and the Hamiltonian variables
according to \eqref{diagnos-relations} is
\begin{align}
\begin{split}
\Pi &:= \frac{\delta\ell}{\delta\Omega} = \frac{1}{\beta}
\mathbb{P}_{\perp A}\left(\Omega - \frac{1}{\alpha}\Gamma\right),
\\
\Xi &:= -\, \frac{\delta\ell}{\delta\Gamma} = \frac{1}{\alpha}\Pi
+ \beta A
\end{split}
\label{PiXi-SO3}
\end{align}

where the operation $\mathbb{P}_{\perp A}$ projects out components
along $A$. The corresponding Lagrangian is found from the Legendre transformation, as
\begin{align}
\ell (\Omega,\Gamma) = \int \Pi\cdot \Omega \, ds -
h(\Pi,\Gamma)\,. \label{lag-SO3}
\end{align}


\begin{remark}\rm
Integrability of the P-chiral $SO(3)$ G-Strand system was
discussed previously in \cite{HoIvPe2012}. However, the integrable
system \eqref{LP-Ham-struct-SO3} for $SO(3)$ G-Strands found here
differs from the one studied in \cite{HoIvPe2012}. In particular,
the system \eqref{LP-Ham-struct-SO3} has a different Lax pair,
which is not related to that for the P-chiral model.
\end{remark}

\section{$SO(4)$ G-Strands}\label{SO4-sec}

This section is an exposition of how the various Lie algebraic properties involved in formulating ZCRs from matrix G-Strand equations are applied in the example of $SO(4)$.

The tangent space at the identity of the matrix Lie group $SO(4)$ is the matrix Lie algebra $\mathfrak{so}(4)$ represented by $4\times4$ skew-symmetric matrices. Any $4\times4$ skew-symmetric matrix $\widehat{\Psi}:O^{-1}\dot{O}\in \mathfrak{so}(4)$ may be represented as a linear combination of $4\times4$ basis matrices with vector coefficients $({\Omega},\,{\Gamma})\in\mathbb{R}^3\times\mathbb{R}^3$ in the following equivalent formulas for the angular velocity of rotation in four dimensions
\begin{align}
\begin{split}
\widehat{\Psi}
&=
\left(
\begin{matrix}
 0  & - \Omega_3 &  \Omega_2  & \Gamma_1  \\
\Omega_3  &  0 & - \Omega_1 & \Gamma_2 \\
-\Omega_2  & \Omega_1   & 0    & \Gamma_3 \\
-\Gamma_1  & -\Gamma_2 &  -\Gamma_3  & 0
\end{matrix}
\right)\\
&=
\sum_{a=1}^3\Omega_a\widehat{J}_a+ \Gamma_a\widehat{K}_a
=:
{\Omega\cdot\widehat{J}}
+
{\Gamma\cdot\widehat{K}}
=: (\Omega,\Gamma)
\,.
\end{split}
\label{matrixrep-so4}
\end{align}
These formulas serve to define the notation that will be found useful in the remainder of the paper.
In particular, the Lie algebra $\mathfrak{so}(4)$ may be represented in the six-dimensional basis of $4\times4$ skew-symmetric matrices $\widehat{J}$ and $\widehat{K}$.
The matrices $\widehat{J}_a$ with $a=1,2,3$ embed the basis for $3\times3$ skew-symmetric matrices into the $4\times4$ matrices by adding a row and column of zeros.
The skew matrices $\widehat{K}_a$ with $a=1,2,3$ then extend the $3\times3$ matrix basis to $4\times4$.

\begin{remark}[Commutation relations]\rm
The skew-symmetric matrix basis of $\mathfrak{so}(4)$ given by $\widehat{J}_a,\,\widehat{K}_b$ with $a,b=1,2,3$ in \eqref{matrixrep-so4} satisfies the commutation relations,
\begin{eqnarray*}
\big[\,\widehat{J}_a,\,\widehat{J}_b\,\big]
&=&
\widehat{J}_a\widehat{J}_b
-
\widehat{J}_b\widehat{J}_a
=
\epsilon_{abc}\widehat{J}_c
\,,
\\
\big[\,\widehat{J}_a,\,\widehat{K}_b\,\big]
&=&
\widehat{J}_a\widehat{K}_b
-
\widehat{K}_b\widehat{J}_a
=
\epsilon_{abc}\widehat{K}_c
=
\big[\,\widehat{K}_a,\,\widehat{J}_b\,\big]
\,,
\\
\big[\,\widehat{K}_a,\,\widehat{K}_b\,\big]
&=&
\widehat{K}_a\widehat{K}_b
-
\widehat{K}_b\widehat{K}_a
=
\epsilon_{abc}\widehat{J}_c
\,,
\end{eqnarray*}
which may be verified directly by a series of direct calculations, as $[\,\widehat{J}_1,\,\widehat{J}_2\,]
=\widehat{J}_3$, etc.
\end{remark}

\begin{remark}[Hat map for $4\times4$ skew-symmetric matrices]\rm
The map \eqref{matrixrep-so4} for the $4\times4$ skew matrix $\widehat{\Psi}$ in the skew-symmetric matrix basis of $\mathfrak{so}(4)$ given by $\widehat{J}_a,\,\widehat{K}_b$ provides the $4\times4$ version of the well-known hat map \cite{Ho2011GM2}, written now as
\[
(\,\cdot\,)\widehat{\phantom{u}}:\,\mathbb{R}^3\times\mathbb{R}^3
\mapsto \mathfrak{so}(4)
,
\]
with $(\Omega,\Gamma)\in\mathbb{R}^3\times\mathbb{R}^3$.
\end{remark}

\subsection{Commutator on $\mathfrak{so}(4)$ as an intertwined vector product}
The commutator of $4\times4$ skew-symmetric matrices corresponds to an intertwined vector product, as follows. For any vectors ${\Omega,\,\Gamma,\,\omega,\,\gamma}\in\mathbb{R}^3$, one has \cite{Ho2011GM2}
\begin{eqnarray*}
&&
{\rm ad}_{(\Omega,\Gamma)}(\omega,\gamma)
=
\Big[\,{\Omega\cdot\widehat{J}}
+
{\Gamma\cdot\widehat{K}}
\,,\,
{\omega\cdot\widehat{J}}
+
{\gamma\cdot\widehat{K}}
\,\Big]
\\&&\hspace{5mm}=
\Big({\Omega\times\omega}
+
{\Gamma\times \gamma}
\Big)\cdot
{\widehat{J}}
+
\Big({\Omega\times \gamma}
+
{\Gamma\times\omega}
\Big)\cdot
{\widehat{K}}
\\&&\hspace{5mm}=:
\Big(\Omega\times\omega + \Gamma\times \gamma
\,,\,
\Omega\times \gamma + \Gamma\times\omega\Big)
\,.
\end{eqnarray*}
Likewise, the dual operation ad$^*$ computed by choosing the matrix pairing $\langle\,A,B\,\rangle={\rm tr}(A^TB)$ is related to the vector dot-product pairing in $\mathbb{R}^3\times\mathbb{R}^3$ by
\begin{align}
\Big\langle\,
{\Omega\cdot\widehat{J}}
+
{\Gamma\cdot\widehat{K}}
\,,\,
{\omega\cdot\widehat{J}}
+
{\gamma\cdot\widehat{K}}
\,\Big\rangle
=
{\Omega\cdot \omega}
+
{\Gamma\cdot \gamma}
\,.
\label{dot-prod1}
\end{align}
That is,
\begin{align}
\Big\langle\,
\widehat{J}_a\,,\,\widehat{J}_b
\,\Big\rangle
=
\delta_{ab}
=
\Big\langle\,
\widehat{K}_a\,,\,\widehat{K}_b
\,\Big\rangle
\quad\hbox{and}\quad
\Big\langle\,
\widehat{J}_a\,,\,\widehat{K}_b
\,\Big\rangle
=
0
\,.
\label{dot-prod2}
\end{align}
One finds ad$^*$ from its definition,
\begin{align}
\Big\langle
{\rm ad}^*_{(\omega,\gamma)}(\pi,\xi), (\Omega,\Gamma)
\Big\rangle
=
\Big\langle
(\pi,\xi),\, {\rm ad}_{(\omega,\gamma)}(\Omega,\Gamma)
\Big\rangle
.
\label{ad-starSO4}
\end{align}
Then, after rearranging the pairing in \eqref{ad-starSO4} using its definition in terms of the dot-product of vectors in \eqref{dot-prod1} one finds the vector representation
\begin{align}
\begin{split}
{\rm ad}^*_{(\omega,\gamma)}(\pi,\xi)
&=
-\,( \omega\times\pi + \gamma\times\xi
\,,\, \omega\times\xi + \gamma\times\pi)
\\
&= -\,{\rm ad}_{(\omega,\gamma)}(\pi,\xi)
\end{split}
\end{align}
Hence, the hat map for $4\times4$ skew-symmetric matrices allows ad and ad$^*$ to be identified for $\mathfrak{so}(4)$, because ${\rm ad}^*_{(\omega,\gamma)}=-\,{\rm ad}_{(\omega,\gamma)}$. This property will be useful when we seek the zero curvature representation of the $SO(4)$ G-Strand equations.

\subsection{Euler--Poincar\'e equation on $\mathfrak{so}(4)^*$}
For a skew-symmetric $4\times4$ matrix
\[
\Phi=O^{-1}\delta O(t)
=\xi\cdot\widehat{J}
+\eta\cdot\widehat{K}\in \mathfrak{so}(4)
\,,\]
Hamilton's principle $\delta S=0$ for
$
S=\int_a^b \ell({\Psi})\,dt
$
with
\[
\Psi=O^{-1}\dot{O}(t)=\Omega\cdot\widehat{J}+\Gamma\cdot\widehat{K}\in \mathfrak{so}(4)
\]
leads to \cite{HoMaRa1998, Ho2011GM2}
\begin{eqnarray*}
\delta S
&=&
\int_a^b
\Big\langle \frac{\delta\ell}{\delta {\Psi}}
\,,\,\delta {\Psi}\Big\rangle
\,dt
=
\int_a^b
\Big\langle \frac{\delta\ell}{\delta {\Psi}}
\,,\,\dot{\Phi}+{\rm ad}_{\Psi}\Phi\Big\rangle
\,dt
\,,
\end{eqnarray*}
where $\delta\Psi=\dot{\Phi}+{\rm ad}_{\Psi}\Phi$ and
\begin{eqnarray*}
{\rm ad}_{\Psi}\Phi
&=&
[\Psi,\,\Phi]
=
\big[\,\Omega\cdot\widehat{J}+\Gamma\cdot\widehat{K}
,\,\xi\cdot\widehat{J}+\eta\cdot\widehat{K}\,\big]
\\
&=&
\Big(\Omega\times\xi
+
\Gamma\times \eta
\Big)\cdot
\widehat{J}
+
\Big(\Omega\times \eta
+
\Gamma\times\xi
\Big)\cdot
\widehat{K}
\,.
\end{eqnarray*}
Inserting these relations for $\delta\Psi$ and ${\rm ad}_{\Psi}\Phi$ into the variational principle yields
\begin{eqnarray*}
0 = \delta S
&=&
\int_a^b
\Big\langle -\,\frac{d}{dt}\frac{\delta\ell}{\delta {\Psi}}
\,,\,\Phi
\Big\rangle
+
\Big\langle
\frac{\delta\ell}{\delta {\Psi}}
\,,\,
{\rm ad}_{\Psi}\Phi\Big\rangle
\,dt
\\
&=&
\int_a^b
\bigg\langle
-\,\frac{d}{dt}\frac{\delta\ell}{\delta \Omega}
\cdot \widehat{J}
-\,\frac{d}{dt}
\frac{\delta\ell}{\delta \Gamma}
\cdot \widehat{K}
\,,\,\xi\cdot\widehat{J}+\eta\cdot\widehat{K}
\bigg\rangle
\,dt
\\
&&\hspace{1cm}+\
\int_a^b
\bigg\langle
\frac{\delta\ell}{\delta \Omega}
\cdot \widehat{J}
+
\frac{\delta\ell}{\delta \Gamma}
\cdot \widehat{K}
\,,\,
\\
&&\hspace{2cm}
\Big(\Omega\times\xi
+
\Gamma\times \eta
\Big)\cdot
\widehat{J}
+
\Big(\Omega\times \eta
+
\Gamma\times\xi
\Big)\cdot
\widehat{K}
\bigg\rangle
\,dt
\\
&=&
\int_a^b
\Big(
-\,\frac{d}{dt}\frac{\delta\ell}{\delta \Omega}
+
\frac{\delta\ell}{\delta \Omega}\times \Omega
+
\frac{\delta\ell}{\delta \Gamma} \times \Gamma
\Big)
\cdot \xi
\\
&&\hspace{1cm}+\
\Big(
-\,\frac{d}{dt}\frac{\delta\ell}{\delta \Gamma}
+
\frac{\delta\ell}{\delta \Gamma}\times \Omega
+
\frac{\delta\ell}{\delta \Omega} \times \Gamma
\Big)
\cdot \eta
\,dt
\,.
\end{eqnarray*}
Hence,  Hamilton's principle $\delta S=0$ yields the following two equations comprising the Euler--Poincar\'e system,
\begin{align}
\begin{split}
\frac{d}{d t} \frac{\delta\ell}{\delta \Omega}
&= \frac{\delta\ell}{\delta \Omega}\times \Omega
+
\frac{\delta\ell}{\delta \Gamma} \times \Gamma
\,,\\
\frac{d}{d t}  \frac{\delta\ell}{\delta \Gamma}
&=
\frac{\delta\ell}{\delta \Gamma}\times \Omega
+
\frac{\delta\ell}{\delta \Omega} \times \Gamma
\,.
\end{split}
\label{so4EP}
\end{align}

\noindent
These are the $\widehat{J},\,\widehat{K}$ basis components of the Euler--Poincar\'e equation on $\mathfrak{so}(4)^*$,
\[
\frac{d}{d t}\frac{\delta\ell}{\delta {\Psi}}
= {\rm ad}^*_{\Psi}\frac{\delta\ell}{\delta {\Psi}}
= -\,{\rm ad}_{\Psi}\frac{\delta\ell}{\delta {\Psi}}
\quad\hbox{or}\quad
\frac{d}{d t}
\left(\frac{\delta\ell}{\delta \Omega},\frac{\delta\ell}{\delta \Gamma} \right)
= -\,{\rm ad}_{(\Omega,\Gamma)}
\left(\frac{\delta\ell}{\delta \Omega},\frac{\delta\ell}{\delta \Gamma} \right)
\,,
\]
written using ${\Psi}
=
(\Omega,\Gamma)
=
\Omega\cdot\widehat{J}
+
\Gamma\cdot\widehat{K}$ in the matrix basis defined in \eqref{matrixrep-so4}.

\subsection{Hamiltonian form on $\mathfrak{so}(4)^*$}
Legendre-transforming $\ell(\Omega,\Gamma)$
to the Hamiltonian
\[
h(\Pi,\Xi) = \big\langle
\Pi,\Omega
\big\rangle
-
\ell(\Omega,\Gamma)
\]
yields the pairs
\[
\Pi=\frac{\delta\ell}{\delta \Omega}
\,, \qquad
\Omega=\frac{\delta h}{\delta\Pi}
\,, \quad\hbox{and}\quad
\Xi=-\,\frac{\delta\ell}{\delta \Gamma}
\,, \qquad
\Gamma =-\,\frac{\delta h}{\delta\Xi}
\,.
\]
Hence, the Euler--Poincar\'e system \eqref{so4EP} may be expressed in Hamiltonian form as
\begin{equation}
\frac{d}{d t}
    \begin{bmatrix}
   {\Pi}
    \\
   {\Xi}
    \end{bmatrix}
=\begin{bmatrix}
{\Pi}\times  &  {\Xi}\times\\
{\Xi}\times  &   {\Pi}
\times
\end{bmatrix}
    \begin{bmatrix}
   \delta h/\delta{\Pi} = \Omega \\
   \delta h/\delta{\Xi} = \Gamma
    \end{bmatrix}
    .
\label{LP-Ham-matrix-so4}
\end{equation}

The corresponding Lie--Poisson bracket is given by
\begin{align}
\{f,\,h\}
&=
-\,{\Pi}\cdot \bigg(
\frac{\delta f}{\delta\Pi}\times\frac{\delta h}{\delta\Pi}
+
\frac{\delta f}{\delta\Xi}\times\frac{\delta h}{\delta\Xi}
\bigg)\\
&\hspace{1cm}
-\,{\Xi}\cdot \bigg(
\frac{\delta f}{\delta\Pi}\times\frac{\delta h}{\delta\Xi}
-
\frac{\delta h}{\delta\Pi}\times\frac{\delta f}{\delta\Xi}
\bigg).
\label{LP-brkt-so4}
\end{align}
\begin{remark}\rm
The Hamiltonian matrix in \eqref{LP-brkt-so4} for this Lie--Poisson bracket has two null eigenvectors for the variational derivatives of $C_1=|\Pi|^2+|\Xi|^2$ and $C_2=\Pi\cdot\Xi$. The functions $C_1,\,C_2$ are called the \emph{Casimirs} of the Lie--Poisson bracket on $\mathfrak{so}(4)^*$. That is, $\{C_1,\,H\}=0=\{C_2,\,H\}$ for every Hamiltonian $H(\Pi,\,\Xi)$.
\end{remark}

\subsection{G-Strand equations on $SO(4)$}\label{GS-SO4-sec}
Let us define the following left-invariant connection form on $\mathfrak{so}(4)$,
\[
O^{-1}dO
= O^{-1}\prt_sO\,ds + O^{-1}\prt_tO\,dt
= (\Omega\cdot\widehat{J} + \Gamma\cdot\widehat{K})ds
+
(\omega\cdot\widehat{J} + \gamma\cdot\widehat{K})dt
=
\widehat{\Upsilon}
\]
where $\prt_s$ denotes the partial derivative with respect to $s$ and $\prt_t$ is the partial derivative with respect to $t$. Inserting this into Cartan's second structure equation
\begin{equation}
d\widehat{\Upsilon}+\frac12 \widehat{\Upsilon}\wedge \widehat{\Upsilon}=0
\,,
\label{Cartan2}
\end{equation}
and taking the $\widehat{J}$ and $\widehat{K}$ components yields
\begin{align}
\begin{split}
\prt_t\Gamma &=  ( \prt_s + \Omega\times)\gamma  + {\Gamma\times \omega}
\,,\\
\prt_t\Omega  &= ( \prt_s + \Omega\times)\omega + {\Gamma\times \gamma}
\,.
\end{split}
\label{ZCR1}
\end{align}

By adapting equation \eqref{so4EP} to the G-Strand case, we find
\begin{align}
\begin{split}
\prt_t \frac{\delta\ell}{\delta {\omega}}
&+{\omega}\times\frac{\delta\ell}{\delta {\omega}}
-\,
\frac{\delta\ell}{\delta \gamma} \times \gamma
=
-\,\prt_s \frac{\delta\ell}{\delta {\Omega}}
- {\Omega}\times\frac{\delta\ell}{\delta {\Omega}}
- \frac{\delta\ell}{\delta \Gamma} \times \Gamma
\,,\\
\prt_t  \frac{\delta\ell}{\delta \gamma}
&+{\omega}\times\frac{\delta\ell}{\delta \gamma}
-
\frac{\delta\ell}{\delta {\omega}} \times \gamma
=
- \,\prt_s  \frac{\delta\ell}{\delta \Gamma}
-\Omega\times\frac{\delta\ell}{\delta \Gamma}
+
\frac{\delta\ell}{\delta \Omega} \times \Gamma
\,.
\end{split}
\label{so4GS}
\end{align}
The Legendre-transform will introduce the variables
\[
\dede{l}{{\omega}} =: \pi
,\quad
\dede{l}{\gamma} =: \xi
,\quad
\dede{l}{\Omega} =: -\Pi
,\quad
\dede{l}{\Gamma}  =: -\Xi
\]
and the relations for the dual variables
\begin{equation} \label{dual var so4}
\dede{h}{\pi} = {\omega}
\,,\quad
\dede{h}{\xi} = \gamma
,\quad
\dede{h}{\Omega} = \Pi
,\quad
\dede{h}{\Gamma}  =  \Xi
\,.
\end{equation}
Therefore, let's write equations \eqref{ZCR1}--\eqref{so4GS} in this notation as
\begin{align}
\begin{split}
\prt_t\pi
&=
(\prt_s + \Omega\times)\Pi
+ \Gamma\times\Xi
+\pi\times\omega + \xi \times \gamma
\,,\\
\prt_t \xi
&=
(\prt_s  + \Omega\times)\Xi
+ \Gamma\times\Pi
+ \pi \times \gamma + \xi \times \omega
\,,\\
\prt_t\Gamma &=  ( \prt_s + \Omega\times)\gamma  + {\Gamma\times \omega}
\,,\\
\prt_t\Omega  &= ( \prt_s + \Omega\times)\omega + {\Gamma\times \gamma}
\,.
\end{split}
\label{so4-GS}
\end{align}

\subsection{Hamiltonian structure for G-Strands on $SO(4)$}\label{Gstrand-SO4}
The Hamiltonian structure for G-Strands on $SO(4)$ is, cf. \eqref{LP-Ham-matrix-so4},
\begin{equation}
\frac{\partial}{\partial t}
\begin{bmatrix}
    \pi
    \\
    \xi
    \\
    \Gamma
    \\
    \Omega
    \end{bmatrix}
\!=\!
 \begin{bmatrix}
   \pi\times  &   \xi\times  &  \Gamma\times  &   (\partial_s + \Omega\times)
   \\
  \xi\times  &  \pi\times  &  (\partial_s + \Omega\times)  &  \Gamma\times
   \\
 \Gamma\times  &  (\partial_s + \Omega\times) &  0  &  0
 \\
   (\partial_s + \Omega\times)  &  \Gamma\times  &  0  &  0
   \end{bmatrix}
   \begin{bmatrix}
   \delta h/\delta\pi=\omega\\
   \delta h/\delta \xi=\gamma\\
   \delta h/\delta\Gamma=\Xi\\
      \delta h/\delta \Omega=\Pi
   \end{bmatrix}.
\label{LP-Ham-struct-vec-s04}
\end{equation}
This Hamiltonian matrix yields a Lie--Poisson bracket defined on the dual of the semidirect-product Lie algebra $\mathfrak{so}(4)\circledS\,(\mathbb{R}^3\times\mathbb{R}^3)$ with a two-cocycle given by $\partial_s$. Dual variables are $(\pi,\xi)$ dual to $\mathfrak{so}(4)$ and $(\Gamma,\Omega)$ dual to $(\mathbb{R}^3\times\mathbb{R}^3)$.
For more discussion of the properties of this type of Hamiltonian matrix, see \cite{GBRa2009,GiHoKu1982,MaRa1999}.

The Euler--Poincar\'e equations in this system may be written individually in matrix commutator bracket notation as
\begin{align}
\begin{split}
\partial_t (\pi,\xi) &= \partial_s (\Pi,\Xi) + {\rm ad}^*_{(\omega, \gamma)}(\pi,\xi)
 - {\rm ad}^*_{(\Omega, \Gamma)}(\Pi,\Xi)
\\
&= \partial_s (\Pi,\Xi) -\,{\rm ad}_{(\omega, \gamma)}(\pi,\xi)
 + {\rm ad}_{(\Omega, \Gamma)}(\Pi,\Xi)
\\
&= \partial_s (\Pi,\Xi) -\,[(\omega, \gamma),(\pi,\xi)]
 + [(\Omega, \Gamma),(\Pi,\Xi)]
\\
 \partial_t (\Omega, \Gamma) &=  \partial_s (\omega, \gamma)
  - [(\omega, \gamma),(\Omega, \Gamma)]
\,.\end{split}
\label{ad-adstar-form}
\end{align}

\subsection{The ZCR for the $SO(4)$ G-Strand}

\begin{theorem}\rm [ZCR formulation for the $SO(4)$ G-Strand] \label{ZCR-thm-SO4}$\,$\\
The $SO(4)$ G-Strand admits a ZCR of the form
\begin{equation}
\partial_t L - \partial_s N = [L , N]\,.
\label{ZCR-so4}
\end{equation}
The matrices for $L$ and $N$ for G-Strands on $SO(4)$ each form a quadratic polynomial in a constant spectral parameter $\lambda$ in which $L$ contains the prognostic variables $(\pi,\xi)$ and $(\Omega, \Gamma)$ in \eqref{LP-Ham-struct-vec-s04} and $N$ contains the diagnostic variables $(\Pi,\Xi)$ and $(\omega, \gamma)$, as
\begin{align}
\begin{split}
 L &= \lambda^2 A  + \lambda (\pi,\xi) +  (\Omega, \Gamma)
\,,\\
 N &= \lambda^2 B + \lambda (\Pi,\Xi) +  (\omega, \gamma)
\,.
\end{split}
\label{LNso4-defs}
\end{align}
\end{theorem}
\begin{proof}
One may substitute these expressions for $L$ and $N$ into the ZCR
in \eqref{ZCR-so4}, then collect coefficients of powers of
$\lambda$ as before to find the following system, cf. equation
\eqref{constraint-rels-SO3},  
\begin{align}
\begin{split}
&\lambda^4:  [A,B]  = 0
\\
&\lambda^3: {\rm ad}_{A}(\Pi,\Xi) = {\rm ad}_B(\pi,\xi)
\\
&\lambda^2:  {\rm ad}_A(\omega, \gamma) = {\rm
ad}_{(\pi,\xi)}(\Pi,\Xi) - {\rm ad}_B(\Omega,\Gamma)
\\
&\lambda:
\partial_t (\pi,\xi) -
\partial_s (\Pi,\Xi) = [(\pi,\xi),(\omega, \gamma)]
 + [(\Omega, \Gamma),(\Pi,\Xi)]
\\
&\lambda^0:
 \partial_t (\Omega, \Gamma) -  \partial_s (\omega, \gamma)
  = [(\Omega, \Gamma),(\omega, \gamma)]
\end{split}
\label{ZCR-extratermO4}
\end{align}                            
The last two equations recover the system \eqref{ad-adstar-form}
and the first three equations have been solved analogously to the
system \eqref{constraint-rels}. The results define the relations
among prognostic and diagnostic variables for which the $SO(4)$
G-Strand equations may be written in ZCR form and, thus, will be
completely integrable.

We present here two integrable constructions.

{\bf Example 1}. We can follow the $SO(3)$ example taking
$A=B=(a_1,a_2).$ The diagnostic relations are similar to
\eqref{constraint-rels}:
\begin{equation}
\begin{split}
(\Pi,\Xi)=&(\pi,\xi)+\beta(a_1,a_2) , \\
(\omega, \gamma)=&(\Omega, \Gamma)+\beta(\pi,\xi)+\sigma(a_1,a_2),
\end{split}
\end{equation} where $\beta, \sigma$ are yet undetermined scalars. Then we can write a Hamiltonian,
analogous to (\ref{Hamilt so3 general})
\begin{equation}
h=\int\left(\pi\cdot\Omega+ \xi \cdot \Gamma +(\mu r + \nu)\left(
\frac{\pi^2 +\xi^2}{2}+a_1\cdot\Omega + a_2 \cdot \Gamma
\right)+f(r)\right) ds,
\end{equation}
with $\mu, \nu$ arbitrary constants, $f$ - an arbitrary function
of one argument and $$r=a_1\cdot\pi + a_2 \cdot \xi.$$ From
(\ref{dual var so4}) we determine
\begin{equation}
\begin{split}
\beta =&\mu (a_1\cdot\pi +
a_2 \cdot \xi) + \nu = \mu r + \nu, \\
\sigma =&\mu\left( \frac{\pi^2 +\xi^2}{2}+a_1\cdot\Omega + a_2
\cdot \Gamma \right)+ f'(r).
\end{split}
\end{equation}

{\bf Example 2}. The commutator between $A$ and $B$ in the
$\lambda^4$ relation can be made to vanish by making the choice
$A=(a,b)$, $B=(b,a)$ where $a,b$ are constant 3-dimensional
vectors, $a\times b \neq 0$. The diagnostic relations among the
variables in \eqref{ZCR-extratermO4} can then be resolved as
\begin{equation}
\begin{split}
\Pi=& \xi+\nu a \,,\\
\Xi=&\pi+\nu b \,,\\
\omega=& \Gamma + \nu \pi+ \mu a + \sigma b\,,\\
\gamma=& \Omega + \nu \xi + \sigma a+ \mu b\,,
\end{split}
\label{var-der}
\end{equation} where $\nu, \mu, \sigma $ are arbitrary constants.

The Lie-Poisson equations \eqref{LP-Ham-struct-vec-s04} may be expanded as
\begin{equation}
\begin{split}
\pi_t-\xi_s=&\nu(\Omega \times a+\Gamma \times b) +\pi\times\alpha
+\xi\times\beta \,,\\
\xi_t-\pi_s=&\nu(\Omega \times b+\Gamma \times a)+
\pi\times\beta+\xi\times\alpha \,,\\
\Omega_t - \Gamma_s- \nu \pi_s=& \nu
(\Omega\times\pi+\Gamma\times\xi) + \Omega\times\alpha+
\Gamma\times\beta
\,,\\
\Gamma_t - \Omega_s- \nu \xi_s=& \nu
(\Omega\times\xi+\Gamma\times\pi) + \Omega\times\beta+
\Gamma\times\alpha\,.
\end{split}
\end{equation}

where $\alpha:=\mu a + \sigma b$ and $\beta: = \sigma a +\mu b$
are constant vectors. The corresponding Hamiltonian may then be
obtained from \eqref{var-der} as
\begin{equation}
h=\int \left(\frac{\nu}{2} \pi^2 +\frac{\nu}{2}\xi^2 + \pi \cdot
\Gamma + \xi \cdot \Omega + \alpha\cdot \pi + \beta \cdot \xi +
\nu(a\cdot\Omega+ b\cdot \Gamma) \right) ds\,.
\end{equation}

\end{proof}


\section{$SE(3)$ G-Strands}\label{SE3-sec}

\subsection{$SE(3)$ dynamics of a filament}
The motion of a filament represented by a space curve in $\mathbb{R}^3$ may be defined by a smooth time-dependent map
\[
\mathbf{c}(t,s):\mathbb{R}\times [0,1]\to \mathbb{R}^3
,\]
where $t$ is time and $s$ is arclength. This motion may be represented in terms of the action of the Lie group $SE(3)$ on its Lie algebra $\mathfrak{se}(3)\simeq\mathbb{R}^3$, as
\begin{align}
g(t,s)\mathbf{c}(0,s)
=\mathbf{c}(t,s),
\quad\hbox{for all}\quad
s\in [0,1]
,
\label{GS-rep}
\end{align}
where $g(t,s)\in SE(3)$ is a real-valued map  $g:\,\mathbb{R}\times[0,1]\to SE(3)$, and $\mathbf{c}(0,s)$ is the initial spatial configuration of the curve $\mathbf{c}(t,s)$. The relation \eqref{GS-rep} lifts the problem from the space of three-dimensional curves to the G-Strand map to $SE(3)$.

To find the Euler--Poincar\'e equations of motion represented this way for space curves, or filaments, we consider Hamilton's principle $\delta S=0$ for a left-invariant Lagrangian,
\begin{eqnarray}
S=\int_a^b \!\!\!\int_{-\infty}^\infty\!\! \ell(\Omega,\Xi)\,ds\,dt
\,,
\end{eqnarray}
with $\ell: \mathfrak{se}(3)\times(\mathbb{R}^3\times\mathbb{R}^3)\to\mathbb{R}$ and the following definitions of the tangent vectors
$\Omega=(\bW,\V)\in\mathfrak{se}(3)$ and $\Xi=(\bOm,\bGam)\in(\mathbb{R}^3\times\mathbb{R}^3)$,
\begin{eqnarray}
\Omega(t,s)=g^{-1}\partial_t g(t,s)
\quad\hbox{and}\quad
\Xi(t,s)=g^{-1}\partial_s g(t,s)
\,,
\end{eqnarray}
where $g(t,s)\in SE(3)$ is a real-valued map
$g:\,\mathbb{R}\times\mathbb{R}\to SE(3)$.

These tangent vectors have a $4\times4$ matrix representation in the form,
\begin{align}
\begin{split}
\widehat{\Xi} = (\bOm,\bGam)
&=
\left(
\begin{matrix}
 0  & - \Omega_3 &  \Omega_2  & \Gamma_1  \\
\Omega_3  &  0 & - \Omega_1   & \Gamma_2 \\
-\Omega_2  & \Omega_1   & 0    & \Gamma_3 \\
0  & 0 &  0  & 0
\end{matrix}
\right)\\
&=
\Omega_a\widehat{J}_a+ \Gamma_b\widehat{K}_b
\,.
\end{split}
\label{matrixrep-se3}
\end{align}
This $4\times4$ matrix represents the velocities of rotations and translations in three dimensions. It may be compared with the $\mathfrak{so}(4)$ matrix \eqref{matrixrep-so4}. The two treatments are similar in appearance, but of course their differences are important.

Next, we apply the Euler--Poincar\'e procedure and incorporate the partial spatial derivative in the definition of $\Xi(t,s)=g^{-1}\partial_s g(t,s)\in \mathfrak{se}(3)$ by taking the following steps.

\begin{enumerate} [(a)]
\item
{\bf Compatibility equation}
The auxiliary equation for the evolution of the components $\Xi=(\bOm,\bGam)\in\mathbb{R}^3\times\mathbb{R}^3$ of $\Xi(t,s):\,\mathbb{R}\times[0,1]\to \mathfrak{se}(3)$ may be obtained from its definition and the equality of cross derivatives in $t$ and $s$, as follows.

We take the difference of the two equations for the partial derivatives
\begin{eqnarray*}
\partial_t\Xi(t,s)
&=&
-\big(g^{-1}\partial_tg\big)\big( g^{-1}\partial_s g\big)
+
g^{-1}\partial_t\partial_s g(t,s)
\,,\\
\partial_s\Omega(t,s)
&=&
-\big(g^{-1}\partial_s g\big) \big(g^{-1}\partial_t g\big)
+
g^{-1}\partial_s\partial_t g(t,s)
\,.
\end{eqnarray*}
Taking the difference between these two equations and invoking equality of cross derivatives implies that $\Xi$ evolves by the ad operation on $\mathfrak{se}(3)$.%
\footnote{The variational derivative of $\Omega $ satisfies a similar equation.}
Namely,
\begin{equation}
\partial_t\Xi(t,s) - \partial_s\Omega(t,s)
= \Xi \,\Omega - \Omega\,\Xi
= [\Xi,\, \Omega]
=: -  {\rm ad}_\Omega\Xi
\,.
\label{aux-eqn}
\end{equation}
In terms of vectors in $\mathbb{R}^3\times\mathbb{R}^3$, the ad-operation on $\mathfrak{se}(3)$ is given by
\begin{eqnarray}
{\rm ad}_\Omega\Xi
=
{\rm ad}_{(\bW\,,\,\V)}(\bOm,\,\bGam)
=
\Big(\bW\times\bOm,\,\bW\times \bGam-\bOm\times\V\Big)
\,.
\label{ad-op-se3}
\end{eqnarray}
Hence, the auxiliary equation \eqref{aux-eqn} for $\Omega=(\bW\,,\,\V)$ and $\Xi=(\bOm,\,\bGam)$ in $\mathbb{R}^3\times\mathbb{R}^3$ vector form may be written equivalently as
\begin{equation}
\partial_t(\bOm,\, \bGam)
-
\partial_s(\bW,\,\V)
=
-\Big(\bW\times\bOm,\,\bW\times \bGam-\bOm\times\V\Big)
\,.
\label{Compat-eqn}
\end{equation}
\item
{\bf The ${\ad}^*$ action of $\mathfrak{se}(3)$ on its dual in the cross pairing.}
We shall compute the ${\ad}^*$ action of the Lie algebra $\mathfrak{se}(3)$ on its dual $\mathfrak{se}(3)^*$ by using the following \emph{cross pairing}, defined as \cite{Ra1982}
\begin{equation}
\Big\langle\!\!\Big\langle
(\bPi,\M), (\bOm,\bGam)
\Big\rangle\!\!\Big\rangle
= \bPi\cdot\bGam + \M\cdot\bOm
\label{cross-pairing}
\end{equation}
In terms of the cross pairing \eqref{cross-pairing}, one computes the ${\ad}^*$ action as
\begin{align}
\begin{split}
\Big\langle\!\!\Big\langle
{\rm ad}^*_{(\W\,,\,\V)} (\bPi\,,\,\M)\,,\,
(\bOm\,,\,\bGam)
\Big\rangle\!\!\Big\rangle
&=
\Big\langle\!\!\Big\langle
(\bPi\,,\,\M)\,,\,
{\rm ad}_{(\W\,,\,\V)}(\bOm\,,\,\bGam)
\Big\rangle\!\!\Big\rangle
\\
&\hspace{-2cm}=
\Big\langle\!\!\Big\langle
(\bPi\,,\,\M)\,,\,
\Big(\bW\times\bOm,\,\bW\times \bGam-\bOm\times\V\Big)
\Big\rangle\!\!\Big\rangle
\\
&\hspace{-2cm}=
\bPi\cdot\Big(\bW\times \bGam-\bOm\times\V\Big)
+ \M\cdot(\bW\times\bOm)
\\
&\hspace{-2cm}=
\Big\langle\!\!\Big\langle
\Big(\bPi\times\bW,\,\W\times \M-\V\times\bPi\Big)
\,,\,
(\bOm\,,\,\bGam)
\Big\rangle\!\!\Big\rangle
.
\end{split}
\label{adstar-se3}
\end{align}
Consequently, we find a linear relation between the ad and ad$^*$ operations on $\mathfrak{se}(3)$,
\begin{align}
\begin{split}
{\rm ad}^*_\Omega\Xi
=
{\rm ad}^*_{(\W\,,\,\V)}(\bPi\,,\,\M)
&=
-\,\Big(\bW\times\bPi,\,\M\times \W+\V\times\bPi\Big)
\\
&=
-\,{\rm ad}_{(\W\,,\,\V)}(\bPi\,,\,\M)
=
-\, {\rm ad}_\Omega\Xi
\end{split}
\label{adstar-def-se3}
\end{align}
Hence, the cross pairing \eqref{cross-pairing} enables ad$^*$ and ad to be identified for $\mathfrak{se}(3)$ by setting ${\rm ad}^*_{(\W\,,\,\V)}=-\,{\rm ad}_{(\W\,,\,\V)}$ as in \eqref{addagger-adstar}. This is the key property \eqref{ad-adstar} for the $SE(3)$ G-Strand equations to admit a zero curvature representation.

\item {\bf Applying the Euler--Poincar\'e theorem.} Using the
Euler--Poincar\'e theorem for left-invariant Lagrangians will
allow us to obtain the equations of motion for the momentum
$\Pi:\,\mathbb{R}\times[0,1]\to \mathfrak{se}(3)^*\simeq
\mathbb{R}^3\times\mathbb{R}^3$, as follows, where we denote
\[
\Pi := \frac{\delta \ell}{\delta \Omega} =\left( \frac{\delta
\ell}{\delta \bW} ,\, \frac{\delta \ell}{\delta \V} \right) =:
(\bPi,\M) \,.
\]
In deriving the Euler--Poincar\'e equations, we will use the $L^2$ pairing defined by spatial integration of the $\mathbb{R}^3\times\mathbb{R}^3$ cross pairing in \eqref{cross-pairing},
\[
\Big<\Pi\,,\,\Xi \Big>
=
\Big<(\bPi,\M)\,,\,(\bOm,\V)  \Big>
=
\int (\bPi\cdot\V + \M\cdot\bOm)\,ds
\,.
\]
We will assume homogeneous endpoint and boundary conditions on $\Omega(t,s)$, $\Xi(t,s)$ and on the variation
\[
\widetilde{\Xi}
:= (\widetilde{\bOm},\,\widetilde{\V})
=\tilde{g}^{-1}\delta \tilde{g}(t,s)\in \mathfrak{se}(3)
\]
when integrating by parts.

As in the previous section, Hamilton's principle $\delta S=0$ yields the Euler--Poincar\'e part of the G-Strand equations for $\delta\ell/\delta\Omega\in\mathfrak{se}(3)^*$,
\begin{equation}
\frac{\partial}{\partial t} \frac{\delta\ell}{\delta \Omega}
= {\rm ad}^*_\Omega\frac{\delta\ell}{\delta \Omega}
- \frac{\partial}{\partial s}  \frac{\delta\ell}{\delta \Xi}
+ {\rm ad}^*_\Xi\frac{\delta\ell}{\delta \Xi}
\,.
\label{2timeEP}
\end{equation}
By identifying ${\rm ad}^*_\Omega=-\,{\rm ad}_\Omega$ for the choice of the cross pairing in \eqref{cross-pairing} and using the definitions $\Omega:=(\bW,\V)$ and $\Xi:=(\bOm,\bGam)$, the previous equation may be written in vector form as
\begin{align}
\begin{split}
\frac{\partial}{\partial t}
&\left(
\frac{\delta\ell}{\delta \bW}, \frac{\delta\ell}{\delta \V}
\right)
+ {\rm ad}_{(\bW,\V)}
\left(
\frac{\delta\ell}{\delta \bW}, \frac{\delta\ell}{\delta \V}
\right)
\\
&+
\frac{\partial}{\partial s} \left(
\frac{\delta\ell}{\delta \bOm}, \frac{\delta\ell}{\delta \bGam}
\right)
+
{\rm ad}_{(\bOm,\bGam)}
\left(
\frac{\delta\ell}{\delta \bOm}, \frac{\delta\ell}{\delta \bGam}
\right)
= 0
\,.
\end{split}
\label{2timeEPa}
\end{align}
On using the definition of ad in equation \eqref{ad-op-se3}, this Euler--Poincar\'e system on $\mathbb{R}^3\times\mathbb{R}^3$ develops into
\begin{align}
\begin{split}
\frac{\partial}{\partial t}
&\left(
\frac{\delta\ell}{\delta \bW}, \frac{\delta\ell}{\delta \V}
\right)
+
\left(
\bW\times \frac{\delta\ell}{\delta \bW},
\bW\times\frac{\delta\ell}{\delta \V}
- \frac{\delta\ell}{\delta \bW}\times\V
\right)
\\
&+
\frac{\partial}{\partial s} \left(
\frac{\delta\ell}{\delta \bOm}, \frac{\delta\ell}{\delta \bGam}
\right)
+
\left(
\bOm\times \frac{\delta\ell}{\delta \bOm},
\bOm\times\frac{\delta\ell}{\delta \bGam}
- \frac{\delta\ell}{\delta \bOm}\times\bGam
\right)
= 0
\,.
\end{split}
\label{2timeEPb}
\end{align}
\begin{remark}\rm
The main effect of using the cross pairing \eqref{cross-pairing} is to change the order of the following variational derivatives,
\[
\left(\bPi=\frac{\delta\ell}{\delta \bW}, \M=\frac{\delta\ell}{\delta \V}\right)\to(\M,\bPi)
\quad\hbox{and}\quad
\left(\bM = - \frac{\delta\ell}{\delta \bOm},\bN=-\frac{\delta\ell}{\delta \bGam}\right)\to(\bN,\bM)
\,.\]
Hence, the Euler--Poincar\'e system in \eqref{2timeEPb} may be written in vector form as
\begin{align}
\begin{split}
\frac{\partial}{\partial t}
&\left(
\M, \bPi
\right)
+
\left(
\bW\times \M,
\bW\times\bPi
- \M\times\V
\right)
\\
&-\,
\frac{\partial}{\partial s} \left(
\bN, \bM
\right)
+
\left(
-\bOm\times \bN,
-\bOm\times\bM
+ \bN\times\bGam
\right)
= 0
\,.
\end{split}
\label{2timeEPc}
\end{align}
\end{remark}

This system of Euler--Poincar\'e equations \eqref{2timeEPc} must be completed by including the compatibility equation \eqref{Compat-eqn}.

\item
{\bf The Legendre transform to the Hamiltonian side.}
We take the Legendre transform of the Lagrangian $\ell(\Omega,\Xi)=\ell((\bW,\V)\,,\,(\bOm,\bGam))$ to obtain the corresponding Hamiltonian. We then differentiate the Hamiltonian to find its partial derivatives and write the Euler--Poincar\'e equation in terms of the momentum variable $\Pi=\delta\ell/\delta\Omega\in \mathfrak{se}(3)^*$.

Legendre transforming in $\Omega$ yields
\begin{eqnarray*}
\frac{\delta l}{\delta \Omega} = \Pi
\,,\quad
\frac{\delta h}{\delta \Pi} = \Omega
\quad\hbox{and}\quad
\frac{\delta h}{\delta \Xi} = -\,\frac{\delta \ell}{\delta \Xi}
\,,
\end{eqnarray*}
where the relation $\Pi - {\delta l/\delta \Omega}=0$ defines $\Pi$.
These derivatives allow one to rewrite the Euler--Poincar\'e equations solely in terms of momentum $\Pi=(\bPi,\,\M)\in \mathfrak{se}(3)^*$ and its dual velocity $\Xi=(\bOm,\, \bGam)\in \mathfrak{se}(3)$ as
\begin{align}
\begin{split}
{\partial_t} \Pi
&= {\rm ad}^*_{\delta h/\delta \Pi}\, \Pi
+ \partial_s \frac{\delta h}{\delta \Xi}
- {\rm ad}^*_\Xi\,\frac{\delta h}{\delta \Xi}
\,,\\
\partial_t \Xi
&= \partial_s\frac{\delta h}{\delta \Pi}
-  {\rm ad}_{\delta h/\delta \Pi}\,\Xi
\,.
\end{split}
\label{GSeqns-se3}
\end{align}
These are the complete G-Strand equations for the Lagrangian $\ell: \mathfrak{se}(3)\to\mathbb{R}$.

\item
We determine the Lie--Poisson bracket implied by the G-Strand equations in terms of the momenta $\Pi=\delta\ell/\delta\Omega$ obtained from the  Legendre transformation, by rearranging the time derivative of a smooth real function $f(\Pi,\Xi):\,\mathfrak{g}^*\times \mathfrak{g}\to \mathbb{R}$.

Assembling equations (\ref{GSeqns-se3}) into Lie--Poisson Hamiltonian form gives, symbolically,%
\footnote{This Hamiltonian matrix is a subset of the Hamiltonian matrix for a perfect complex fluid \cite{Ho2002,GBRa2009}. It also appears in the Lie--Poisson brackets for Yang--Mills fluids \cite{GiHoKu1982} and for spin glasses \cite{HoKu1988}. It is the Hamiltonian structure for a Lie--Poisson bracket on the dual of the semidirect-product Lie algebra $\mathfrak{se}(3)\circledS\,(\mathbb{R}^3\times\mathbb{R}^3)$ with a two-cocycle given by $\partial_s$.}
%
\begin{equation}
\frac{\partial}{\partial t}
    \begin{bmatrix}
    \Pi =(\bPi,\,\M)
    \\
    \Xi =(\bOm,\, \bGam)
    \end{bmatrix}
=
\begin{bmatrix}
  {\rm ad}^\ast_\square \Pi
   &
   (\partial_s - {\rm ad}^*_\Xi)\square
   \\
   (\partial_s - {\rm ad}_\square)\Xi
   & 0
    \end{bmatrix}
    \begin{bmatrix}
   \delta h/\delta\Pi = \Omega = (\W,\V)\\
   \delta h/\delta\Xi = \Gamma = (\bM,\bN)
    \end{bmatrix}
 \label{SMK-HamMatrix}
\end{equation}
where the boxes $\square$ indicate how the ad- and ad$^*$\!-operations occur in the matrix multiplication. For example,
\[
{\rm ad}^\ast_\square \Pi(\delta h/\delta\Pi)= {\rm ad}^\ast_{\delta h/\delta\Pi} \Pi
\,,
\]
so each entry in the matrix acts on its corresponding vector component. Of course, this is shorthand notation, since $\Pi$ and $\Xi$ are $4\times4$ matrices.
\item
The G-Strand equations for $\mathfrak{g}=\mathfrak{se}(3)$ in terms of the $\mathbb{R}^3\times\mathbb{R}^3$ vector operations in \eqref{2timeEPb} may be written in Hamiltonian matrix form as,
%
\begin{equation}
\frac{\partial}{\partial t}
\left[
\begin{array}{c}
    \bPi
    \\
    \M
    \\
   \bGam
    \\
    \bOm
    \end{array}
\right]
=
\left[
\begin{array}{cccc}
 \bPi\times  &   \M\times     &   \bGam\times   &   (\partial_s + \bOm\times)
   \\
   \M \times  &  0  &  (\partial_s + \bOm\times)   &   0
   \\
   \bGam\times  &  (\partial_s + \bOm\times)  &  0  &  0
   \\
   (\partial_s + \bOm\times)  &  0  &  0  &  0
    \end{array}
\right]
\left[
\begin{array}{c}
   \delta h/\delta\bPi =\bW \\
   \delta h/\delta\M =\V \\
    \delta h/\delta \bGam =\bN  \\
    \delta h/\delta\bOm =\bM
    \end{array}
\right]
\label{LP-Ham-struct-SMK-se3}
\end{equation}

This {\bfi Lie--Poisson} Hamiltonian matrix possesses a two-cocycle given by $\partial_s$, as discussed in \cite{EGBHPR2010,SiMaKr1988}. It may be compared with the Hamiltonian matrix for $SO(4)$ G-Strands in equation \eqref{LP-Ham-struct-vec-s04}.
It simplifies for $\mathfrak{g}=\mathfrak{se}(2)$, since in that case the variables $\{\bPi,\bW,\bOm,\delta h/\delta \bOm=\bM\}$ are all normal to the plane, and the others $\{\M,\V,\bGam,\delta h/\delta \bGam=\bN\}$ are all in the plane. This simplification will be discussed in the next section. Basically, for motion of a planar curve the contributions of the vector cross product terms at the corners of the Lie--Poisson form in \eqref{LP-Ham-struct-SMK-se3} do not contribute to planar motion.

\item
The $SE(3)$ G-Strand equations in \eqref{LP-Ham-struct-SMK-se3} may also be written as individual equations, as in \cite{SiMaKr1988}
\begin{align}
\begin{split}
\displaystyle
\prt_t \bPi - \prt_s \bM
& = \bPi\times\bW +\M\times \V + \bGam\times \bN + \bOm\times\bM
,
\\
 \displaystyle
\prt_t \M - \prt_s\bN
& = \M\times \bW + \bOm\times \bN
,
\\
\prt_t \bGam - \prt_s \V
& = \bGam \times \bW + \bOm\times \V
,\\
\partial_t\bOm - \prt_s\bW
& = \bOm\times\bW.
\end{split}
\label{SE3Gstrand-eqns}
\end{align}
\begin{remark}[Relation to the SMK equations of exact rod theory]\rm$\,$\\
The SE(3) G-Strand equations in \eqref{SE3Gstrand-eqns} coincide with the SMK convective representation of exact rod theory studied in \cite{EGBHPR2010,SiMaKr1988}, except for the extra term $\M\times \V$, which vanished in \cite{SiMaKr1988} for the choice of constitutive relations in the Hamiltonian there, although it was still present in the Lie--Poisson bracket.
The SMK equations describe the space-time dependent bending, twisting, coiling, and wave propagation properties of filaments.
\end{remark}

Next, we determine the conditions required for the $SE(3)$ G-Strand equations to admit a ZCR.
Then, in the section afterward, we discuss the simplifications that result under the reduction from $SE(3)$ to $SE(2)$.

\end{enumerate}\bigskip

\subsection{The ZCR for the motion of space curves by $SE(3)$}
{\bf The ZCR for the $SE(3)$ G-Strand.}
To compute the zero curvature representation of the $SE(3)$ G-Strand equations, we first use the equivalence of ${\rm ad}^*$ and $-\,{\rm ad}$ as in \eqref{addagger-adstar} in the bi-invariant cross pairing to rewrite equation \eqref{GSeqns-se3} equivalently as
\begin{align}
\begin{split}
{\partial_t} \Pi
&= -\,{\rm ad}_{\Gamma}\, \Pi
+ \partial_s \Omega
+ {\rm ad}_\Xi\,\Omega
\,,\\
\partial_t \Xi
&= \partial_s \Gamma
-  {\rm ad}_{\Omega}\,\Xi
\,,
\end{split}
\label{GSeqns-se3a}
\end{align}
where $\Pi=(\bPi,\,\M)$, $\Xi:=(\bOm,\bGam)$, $\Omega:=(\bW,\V)$ and $\Gamma = (\bM,\bN)$, in terms of $\mathbb{R}^3\times\mathbb{R}^3$ vectors, as in equation \eqref{SMK-HamMatrix}.

\begin{theorem}\rm [ZCR formulation for the $SE(3)$ G-Strand] \label{ZCR-thm-SE3}$\,$\\
The $SE(3)$ G-Strand admits a ZCR of the form
\begin{equation}
\partial_t L - \partial_s N = [L , N]\,.
\label{ZCR-so4}
\end{equation}
The matrices for $L$ and $N$ for G-Strands on $SE(3)$ each form a quadratic polynomial in a constant spectral parameter $\lambda$  in which $L$ contains the prognostic variables $(\Pi,\Xi)$ and $N$ contains the diagnostic variables $(\Omega,\Gamma)$, as
\begin{align}
\begin{split}
 L &= \lambda^2 A  + \lambda \Pi +  \Xi
\,,\\
 N &= \lambda^2 B + \lambda \Gamma +  \Omega
\,,
\end{split}
\label{LNse3-defs}
\end{align}
in which the constant matrices $A$ and $B$ commute, $[A,B]  = 0$.
\end{theorem}
\begin{proof}

 One may substitute these expressions for $L$ and $N$
into the ZCR in \eqref{ZCR-so4}, then collect coefficients of
powers of $\lambda$ as done earlier for $\mathfrak{so}(4)$ to find
the following system, cf. equation \eqref{constraint-rels-SO3},

\begin{align}
\begin{split}
&\lambda^4:  [A,B]  = 0
\\
&\lambda^3:{\rm ad}_{A} \Gamma ={\rm ad}_B \Pi
\\
&\lambda^2: {\rm ad}_{A}\Omega =  {\rm ad}_B \Xi - {\rm ad}_\Pi
\Gamma
\\
&\lambda:
\partial_t \Pi - \partial_s \Gamma
=
- {\rm ad}_\Omega\Pi + {\rm ad}_\Xi\Gamma
\\
&\lambda^0:
 \partial_t \Xi -  \partial_s \Omega
  =
  - {\rm ad}_\Omega\Xi
\end{split}
\label{ZCR-extratermSE3}
\end{align}
The last two equations recover the system \eqref{GSeqns-se3a} and
the first three equations have been solved using the same method
as for the general system \eqref{constraint-rels}.

Again, we can follow the $SO(3)$ example taking $A=B=(a_1,a_2).$
The diagnostic relations are similar to \eqref{constraint-rels}:
\begin{equation} \label{diagn rel se3}
\begin{split}
(\bW,\V)=&(\bOm,\bGam)+\beta(\bPi,\M)+\gamma(a_1,a_2) , \\
(\bM, \bN)=&(\bPi, \M)+\beta(a_1,a_2),
\end{split}
\end{equation} where $\beta, \gamma$ are yet undetermined scalars. Then we can write a Hamiltonian,
analogous to (\ref{Hamilt so3 general})
\begin{equation}
h=\int\left(\bPi\cdot\bOm+ \M\cdot \bGam +(\mu r + \nu)\left(
\frac{\bPi^2 +\M^2}{2}+a_1\cdot\bOm + a_2 \cdot \bGam
\right)+f(r)\right) ds,
\end{equation}
with $\mu, \nu$ arbitrary constants, $f$ - an arbitrary scalar
function of one argument and $$r=a_1\cdot\bPi + a_2 \cdot \M.$$ By
matching the variational derivatives of $h$ given in
(\ref{LP-Ham-struct-SMK-se3}) to the diagnostic relations
(\ref{diagn rel se3}),  we obtain the unknown scalars:
\begin{equation}
\begin{split}
\beta =&\mu (a_1\cdot\bPi + a_2 \cdot \M) + \nu = \mu r + \nu, \\
\gamma =&\mu\left( \frac{\bPi^2 +\M^2}{2}+a_1\cdot\bOm + a_2 \cdot
\bGam \right)+ f'(r).
\end{split}
\end{equation}
\end{proof}

\begin{remark}\rm 
The Hamiltonian used in the SMK
construction has the form $$ h=\int \left(\M^2 + \bPi\cdot {\bf J}
\bPi + \psi(\bGam, \bOm) \right) ds, $$ where ${\bf J}$ is a
symmetric operator and $\psi$ is a scalar function of its
arguments. The variational derivatives of such a Hamiltonian (cf.
(\ref{LP-Ham-struct-SMK-se3}) ) do not satisfy the diagnostic
relations and thus the SMK model is unlikely to be integrable.
\end{remark}


\subsection{Reduction to $SE(2)$ motion of planar curves}
The special Euclidean group of the plane $SE(2)\simeq SO(2)\,\circledS\, \mathbb{R}^2$ acts on a point in the plane defined by the vector $q=(q_1,\,q_2)^T\in\mathbb{R}^2$, as follows
\[
(R_\theta(t),v(t))(q)
=
\left(
\begin{array}{cc}
R_\theta(t) & v(t) \\
0 & 1
\end{array}
\right)
\left[
\begin{array}{c}
q \\
1
\end{array}
\right]
=
\left[
\begin{array}{c}
R_\theta(t) q + v(t) \\
1
\end{array}
\right]
,
\]
where $v=(v_1,\,v_2)^T\in \mathbb{R}^2$ is a vector in the plane and $R_\theta$ is the $2\times2$ matrix for rotations of vectors in the plane by angle $\theta$ about the normal to the plane $\hat{z}$,
\[
R_\theta
=
\left(
\begin{array}{cc}
\cos\theta & -\sin\theta \\
\sin\theta & \cos\theta
\end{array}
\right)
.
\]
The infinitesimal action is found by taking $\frac{d}{dt}|_{t=0}$ of this action, which yields
\[
\frac{d}{dt}\Big|_{t=0}(R_\theta(t),v(t))(q)
=
\left(
\begin{array}{cc}
-\Omega\hat{z}\times  & V \\
0 & 0
\end{array}
\right)
\left[
\begin{array}{c}
q \\
1
\end{array}
\right]
=
\left[
\begin{array}{c}
-\Omega\hat{z}\times q + V \\
1
\end{array}
\right]
,
\]
where $\Omega=\dot{\theta}(0)$ and $V=\dot{v}(0)$ are elements of $\mathfrak{se}(2)$.

In vector notation, the $\mathfrak{se}(2)$ ad-action  is
\begin{eqnarray*}
{\rm ad}_{(\Omega\,,\,V)}({\tilde{\Omega}}\,,\,{\tilde{V}})
&=&
[(\Omega, V)\,,\,(\tilde{\Omega},\tilde{V})\,]
\\&=&
\Big(
[\Omega\,,\,{\tilde{\Omega}}\,]\,,\,\Omega{\tilde{V}} - {\tilde{\Omega}}V
\Big)
\\&=&
\Big(
0
\,,\,
-\Omega\hat{z}\times\tilde{V} + \tilde{\Omega}\hat{z}\times V
\Big)
\,.
\end{eqnarray*}
The pairing between the Lie algebra $\mathfrak{se}(2)$ and its dual $\mathfrak{se}(2)^*$ is given by the cross pairing of vectors in $\mathbb{R}^3$,
\[
\Big<(\Pi,\mathcal{M})\,,\,(\Omega,V)  \Big>
=
\Pi\cdot V  + \mathcal{M}\cdot \Omega
\,.
\]
For $\mathfrak{g}=\mathfrak{se}(2)$, the vectors $\{\Pi,{\cal W},\Omega,\delta h/\delta \Omega= M\}$ are all normal to the plane, and the vectors $\{{\cal M},{\cal V},\Gamma,\delta h/\delta \Gamma=N\}$ are all in the plane.
Consequently, the Lie--Poisson Hamiltonian form of the G-Strand equations for $\mathfrak{g}=\mathfrak{se}(3)$ in \eqref{LP-Ham-struct-SMK-se3} reduces to the following simpler form,
\begin{equation} \label{LP-Ham-struct-vec-se2}
\frac{\partial}{\partial t}
\left[
\begin{array}{c}
    \Pi
    \\
    \mathcal{M}
    \\
    \Gamma
    \\
    \Omega
    \end{array}
\right]
=
\left[
\begin{array}{cccc}
   0
   &
   \mathcal{M}\times
   &
   \Gamma\times
   &
   \partial_s
   \\
   \mathcal{M}\times
   &
   0
   &
   (\partial_s + \Omega\times)
   &
   0
   \\
   \Gamma\times
   &
   (\partial_s + \Omega\times)
   &
   0
   &
   0
   \\
   \partial_s
   &
   0
   &
   0
   &
   0
    \end{array}
\right]
\left[
\begin{array}{c}
   \delta h/\delta\Pi=W \\
   \delta h/\delta\mathcal{M}= V \\
    \delta h/\delta \Gamma = N \\
    \delta h/\delta\Omega = M
    \end{array}
\right]
\end{equation}
The $SE(2)$ G-Strand equations for the motions of planar curves may then be written out separately as, cf. \eqref{SE3Gstrand-eqns},
\begin{align}
\begin{split}
\displaystyle
\prt_t \Pi - \prt_s M
& = \mathcal{M}\times V + \Gamma\times N
,
\\
 \displaystyle
\prt_t \mathcal{M} - \prt_s N
& = \mathcal{M}\times W + \Omega\times N
,
\\
\prt_t \Gamma - \prt_s V
& = \Gamma \times W + \Omega\times V
,\\
\partial_t\Omega - \prt_s W
& = 0.
\end{split}
\label{planarcurve-eqns}
\end{align}

Of course, the ZCR for these equations is a subset of that in \eqref{LNse3-defs}. The ZCR for two-dimensional motions generated by $\mathfrak{se}(2)$ contains most of the features of the three-dimensional motions generated by $\mathfrak{se}(3)$, but there are only six variables instead of twelve. Therefore, it might be interesting to build intuition for the solutions of G-Strand motions using the ZCR in the case of the $SE(2)$ G-Strands.

\section{Conclusions}\label{Conclus-sec}

We have shown that the G-Strand equations are a rich source of
integrable systems of physical significance. For example, we have
shown that the SE(3) G-Strand equations and the
Simo-Marsden-Krishnaprasad (SMK) equations of exact rod theory in
the convective  representation both describe the motion of a curve
in space under the action of SE(3). The motion of space curves is
a classical problem that has been deeply studied from the
viewpoint of integrable systems \cite{DoSa1994}. Exploring that
connection with the SE(3) G-Strand equations could be a promising
direction for future research, particularly for the integrable
cases.  Of course, the \emph{G-Brane} equations for the action of
SE(3) on smooth embeddings $Emb(M,\mathbb{R}^3)$ for a smooth
manifold $M\subset\mathbb{R}^3$ would also be promising and the
mathematical formulation for G-Branes has recently been discussed
in \cite{FGB2012}. Applications of similar ideas for registration
of planar curves have recently been successful \cite{HoNoVa2013}.
The further application of these ideas for matching space curves
seems like a natural generalization.

{\bf Travelling waves.} Interestingly, the travelling waves for
ZCR systems have contributions from both $s$- and $t$-
derivatives. Consequently,  they do not reduce to the original
integrable ODEs that inspired the present work. Instead the wave
speed introduces a bifurcation parameter, $c$, and the travelling
wave solution becomes singular when $c^2=1$, which is the
characteristic speed of the full hyperbolic PDE system. The
behaviour of the travelling wave solutions of these ZCR systems
will be studied elsewhere.

{\bf Soliton solutions.} The ZCR of the matrix G-Strands puts them
into the realm of integrable Hamiltonian systems. An interesting
direction would be to pursue soliton solutions of the matrix
G-Strand equations derived here. The integrable matrix G-Strands
possess a ZCR arising as commutation of matrix operators in the
depending on a spectral parameter $\lambda$. Finding their soliton
solutions requires solving an inverse spectral problem. The
corresponding inverse problem may be formulated as a nonlinear
Riemann--Hilbert problem on a given contour of the complex
$\lambda$ plane. For more information, see
\cite{FoItsKaNo2006,Its2003,Iv2004,ZaMaNoPi1984} and references
therein. This method for obtaining the soliton solutions of the
ZCR equations derived from G-Strand equations should be a
promising direction for future investigations.

{\bf Wave behaviour of liquid crystals and other complex fluids.} The effects that have been studied experimentally in the field of complex fluids have been mainly dominated by dissipation. However, the equations for the dynamics of these systems without dissipation admit nonlinear waves. These waves could also be experimentally interesting in studies of complex fluids, which we have found here can be formulated naturally as \EP systems \cite{Ho2002,GBRa2009,EGBHPR2010}. This will be the topic of future research.

{\bf Harmonic maps.}
Incidentally, all of the ZCR examples here can be extended to harmonic maps, by replacing $(s,t)\in\mathbb{R}\times\mathbb{R}$ by $(z,\bar{z}\in\mathbb{C})$ and following Uhlenbeck's approach, \cite{Uh1989}.

\subsection*{Acknowledgements}
We are grateful to D. C. P. Ellis, F. Gay-Balmaz, J. E. Marsden,
V. Putkaradze, T. S. Ratiu and C. Tronci for extensive discussions
of this material. We are particularly grateful to T. S. Ratiu for
suggesting the cross pairing \eqref{cross-pairing} for the $SE(3)$
case. Work by DDH was partially supported by Advanced Grant 267382
FCCA from the European Research Council. DDH is also grateful for
hospitality at the Isaac Newton Institute for Mathematical
Sciences, where this paper was finished.

\end{document}